\documentclass[journal]{IEEEtran}
\usepackage{cite}
\usepackage{algpseudocode}
\usepackage[pdftex]{graphicx}
\usepackage{amsmath}
\usepackage{empheq}
\usepackage{array}

\usepackage{caption}
\usepackage{subcaption}
\usepackage{algpseudocode}
\usepackage{algorithm}
\usepackage[font=scriptsize]{caption}
\usepackage{url}
\usepackage{amsmath,amsthm,amssymb}
\usepackage{txfonts}
\newtheorem{theorem}{Theorem}

\usepackage{tabularx,booktabs}
\newcolumntype{C}{>{\centering\arraybackslash}X} 
\usepackage{color}

\begin{document}

\title{Resilient Synchronization of Distributed Multi-agent Systems under Attacks}

\author {Aquib Mustafa,~\IEEEmembership{Student Member,~IEEE}, {Rohollah Moghadam,~\IEEEmembership{Student Member,~IEEE}, Hamidreza Modares,~\IEEEmembership{Senior Member,~IEEE}} }



\maketitle

\begin{abstract}
In this paper, we first address adverse effects of cyber-physical attacks on distributed synchronization of multi-agent systems, by providing conditions under which an attacker can destabilize the underlying network, as well as another set of conditions under which local neighborhood tracking errors of intact agents converge to zero. Based on this analysis, we propose a Kullback-Liebler divergence based criterion in view of which each agent detects its neighbors' misbehavior and, consequently, forms a self-belief about the trustworthiness of the information it receives. Agents continuously update their self-beliefs and communicate them with their neighbors to inform them of the significance of their outgoing information. Moreover, if the self-belief of an agent is low, it forms trust on its neighbors. Agents incorporate their neighbors' self-beliefs and their own trust values on their control protocols to slow down and mitigate attacks. We show that using the proposed resilient approach, an agent discards the information it receives from a neighbor only if its neighbor is compromised, and not solely based on the discrepancy among neighbors' information, which might be caused by legitimate changes, and not attacks. The proposed approach is guaranteed to work under mild connectivity assumptions.
\end{abstract}

\begin{IEEEkeywords}
Distributed control, Resilient Control, Attack Analysis, Multi-agent systems.
\end{IEEEkeywords}

\section{Introduction}
A Distributed Multi-Agent System (DMAS) is collection of dynamical systems or agents that interact with each other over a communication network to achieve coordinated operations and behaviors \cite{Olfati2007,Bullo2009,Basar2016khan,Basar2015}. In the case of synchronization of DMASs, the objective is to guarantee that all agents reach agreement on a common value or trajectory of interest.  Despite their numerous applications in a variety of disciplines, DMASs are vulnerable to attacks, which is one of the main bottleneck that arise in their wide deployment. In contrast to other undesirable inputs, such as disturbances and noises, cyber-physical attacks are intentionally planned to maximize the damage to the overall system or even destabilize it. 

There has been extensive research progress in developing attack detection/identification and mitigation approaches for both spatially distributed systems  \cite{Mo2014,shoukry2013non,Fawzi2014,Pajic2014,Pasqualetti2013,Shoukry2016Event,Mo2015,kvam2019,Persis2015,Gupta17} and DMASs \cite{Pasqualetti2012TAC,Teixeira2010,Amin2013TCST,Weerakkody2017,Feng2017ACC,Sundaram2011,Zeng2014Resilient,Yucelen2017,LeBlanc2017,Basar2017,Basar2017CPS,Basar2017CPSnew,DIBAJI2017,dibaji2019,Rohollah2017CDC,Ghu2019}. Despite tremendous and welcoming progress, most of the mentioned mitigation approaches for DMASs use the discrepancy among agents and their neighbors to detect and mitigate the effect of an attack. However, as shown in this paper, a stealthy attack can make all agents become unstable simultaneously, and thus misguide existing mitigation approaches. Moreover, this discrepancy could be caused by a legitimate change in the state of an agent, and rejecting this useful information can decrease the speed of convergence to the desired consensus and harm connectivity of the network. 

In this paper, we present attack analysis, detection, and mitigation mechanisms for DMASs with linear structures. We show that local neighborhood tracking errors of intact agents converge to zero, regardless of the attack, if the set of eigenvalues of the attacker signal generator dynamics matrix is a subset of the set of eigenvalues of the system dynamics matrix. We call these types of attacks internal model principle (IMP)-based attacks. In spite of convergence to zero of local neighborhood tracking errors, the overall network could be destabilized, and we provide sufficient conditions for this to happen. We then develop attack detectors that identify both IMP-based and non-IMP-based attacks. To detect IMP-based attacks, two local error sequences with folded Gaussian distributions are introduced based on the relative information of the agents. We show that they diverge under an IMP-based attack. A Kullback-Liebler (KL) divergence criterion is then introduced to measure the divergence between these two univariate folded Gaussian distributions, and consequently capture IMP-based attacks. Similarly, since non-IMP based attacks change the statistical properties of the local neighborhood tracking error, to detect non-IMP-based attacks, the KL divergence is employed to measure the discrepancy between the Gaussian distributions of the actual and nominal expected local neighborhood tracking errors. Then, a self-belief value, as a metric capturing the probability of the presence of attacks directly on sensors or actuators of the agent itself or on its neighbors, is presented for each agent by combining these two KL-based detectors. The self-belief indicates the level of trustworthiness of the agent's own outgoing information, and is transmitted to its neighbors. Furthermore, when the self-belief of an agent is low,  the trustworthiness of its incoming information from its neighbors is estimated using a particular notion of trust. Trust for each individual neighbor is developed based on the relative entropy between the neighbor's information and agent's own information. Finally, by incorporating neighbor's self-belief and trust values, we propose modified weighted control protocols to ensure mitigation of both types of attacks. Simulation results included in the paper validate the effectiveness of the approach.

\vspace{-0.3cm}
\section{Preliminaries}
A directed graph (digraph) $\mathcal{G}$ consists of a pair $(\mathcal{V,\,{\mathcal{E}}})$ in which $\mathcal{V}{\text{  =   }}\{ {v_1}, \cdots ,{v_N}\} $ is a set of nodes and $\mathcal{E} \subseteq \mathcal{V} \times \mathcal{V}$ is a set of edges. We denote the directed link (edge) from $v_j$ to $v_i$ by the ordered pair $(v_j, v_i)$. The adjacency matrix is defined as $\mathcal{A}=[{a}_{ij}]$, with ${a_{ij}} > 0$ if $({v_j},{v_i}) \in \mathcal{E}$, and ${a_{ij}} = 0$ otherwise. We assume there are no repeated edges and no self loops, i.e., ${a_{ii}} = 0 \,\,\,\forall i \in \mathcal{{N}}$ with $\mathcal{N}=\left\{ {1,\ldots,N} \right\}.$ The nodes in the set ${\mathcal{N}_i} = \{ {v_j}:({v_j},{v_i}) \in \mathcal{E}\} $ are said to be neighbors of node ${\nu _i}$. The in-degree of ${v_i}$ is the number of edges having ${v_i}$ as
a head. The out-degree of a node ${v_i}$ is the number of edges having ${v_i}$ as a tail.
If the in-degree equals the out-degree for all nodes ${v_i \in V}$ the graph is said to
be balanced. The graph Laplacian matrix is defined as $\mathcal{L}=\mathcal{D}-\mathcal{A}$, where $\mathcal{D}=\mathrm{diag}(d_i)$ is the in-degree matrix, with ${d_i} = \sum\limits_{j \in {N_i}} {{a_{ij}}} $ as the weighted in-degree of node ${\nu _i}$. A node is called as a \textit{root node} if it can reach all other nodes of the digraph $\mathcal{G}$ through a directed path. A leader is a root node with no incoming link. A (directed) tree is a connected digraph where every node except one, called the root, has in-degree equal to one.  A spanning tree of a digraph is a directed tree formed by graph edges, which connects all the nodes of the graph. 

Throughout the paper, we denote the set of integers by $\mathbb{Z}$. The set of integers greater than or equal to some integer $q\in \mathbb{Z}$ is denoted $\mathbb{Z}_{\geqslant q}$. The cardinality of a set $S$ is denoted by $|S|$. $\lambda (A)$ and $\mathrm{tr}(A)$ denote, respectively, the eigenvalues and trace of the matrix $A$. Furthermore, $\lambda_{min}(A)$ represents minimum eigenvalue  of matrix $A$. The Kronecker product of matrices $A$ and $B$ is denoted by $A \otimes B$, and $\mathrm{diag}\left( {{A_1},\ldots,{A_n}} \right)$ represents a block diagonal matrix with matrices ${A_i}$, $\forall \,\,i\in \mathcal{N}$ as its diagonal entries. ${{\mathbf{1}}_N}$ is the $N$-vector of ones and ${{\mathbf{I}}_N}$ is the $N \times N$ identity matrix. $\operatorname{Im} (R)$ and $\ker (R)$ represent, respectively, the range space and the null space of $R$, and $\mathrm{span}({a_1}, \ldots ,{a_n})$ is the set of all linear combinations of the vectors ${a_1}, \ldots ,{a_n}$. A Gaussian distribution with mean $\mu $ and covariance $\Sigma $ is denoted by $\mathcal{N}\left( {\mu ,\Sigma } \right)$. Moreover, $\mathcal{F}\mathcal{N}\left( {\bar \mu ,{{\bar \sigma }^2}} \right)$ represents univariate folded Gaussian distribution with $\bar \mu $ and ${\bar \sigma ^2}$ as mean and variance, respectively \cite{r30}. $\mathbb{E}[.]$ denotes the expectation operator.
\smallskip

\noindent
\textbf{Assumption 1.} The communication graph $\mathcal{G}$ is directed and has a spanning tree.
\smallskip

\noindent
\textbf{Definition 1  \cite{li2014cooperative}-\cite{lewis2013cooperative}.} A square matrix $A \in {\mathbb{R}^{n \times n}}$ is called a singular M-matrix, if all its off-diagonal elements are non-positive and all its eigenvalues have non-negative real parts. \hfill $\square$
\smallskip

\noindent
\textbf{Definition 2  \cite{li2014cooperative}-\cite{lewis2013cooperative}.} A square matrix $A \in {\mathbb{R}^{n \times n}}$ is called a non-singular M-matrix, if all its off-diagonal elements are non-positive and all its eigenvalues have positive real parts. \hfill $\square$

\noindent
\textbf{Lemma 1 \cite{li2014cooperative}-\cite{lewis2013cooperative}.} The graph Laplacian matrix $\mathcal{L}$ of a directed graph $\mathcal{G}$ has at least one zero eigenvalue, and all its nonzero eigenvalues have positive real parts. Zero is a simple eigenvalue of $\mathcal{L}$, if and only if Assumption 1 is satisfied.

\section{Overview of Consensus in DMASs}
In this section, we provide an overview of the consensus problem for leaderless DMAS. Consider a group of $N$ homogeneous agents with linear identical dynamics described by
\begin{equation}\label{eq1}
{\dot {x}_i}(t) = {A}{x_i}(t) + {B}{u_i}(t)\begin{array}{*{20}{c}}
  {}&{\forall\,\, i \in \mathcal{N}} 
\end{array}
\end{equation}
where ${x_i} \in {\mathbb{R}^{{n}}}$ and ${u_i} \in {\mathbb{R}^{{m}}}$ denote, respectively, the state and the control input of agent $i$. The matrices ${A} \in {\mathbb{R}^{{n} \times {n}}}$ and ${B} \in {\mathbb{R}^{{n} \times {m}}}$ are, respectively, the drift dynamics and the input matrix. \par
\smallskip

\noindent
\textbf{Problem 1.} Design local control protocols $u_i$ for all agents $\forall i \in \mathcal{{N}}$  in \eqref{eq1} such that all agents reach consensus or synchronization (agreement) on some common value or trajectory of interest, i.e.,
\begin{equation}\label{eq4}
\begin{array}{*{20}{c}}
  {\mathop {\lim }\limits_{t \to \infty } ||x_j(t)-x_{i}(t)|| = 0}\,\,\,\,\,{\forall i,j \in \mathcal{{N}}}
\end{array}
\end{equation}
\par
\smallskip

\noindent
\textbf{Assumption 2.} The system dynamics matrix $A$ in \eqref{eq1} is assumed to be marginally stable.
\smallskip

\noindent
\textbf{Remark 1.} Based on Assumption $2$, all the eigenvalues of $A$ have
non-positive real part. This is a standard assumption in the literature for consensus or synchronization problems \cite{Huang2012TAC1}. Note that if $A$ is Hurwitz, the synchronization problem has a trivial solution and can be solved by making the dynamics of each agent stable independently. Moreover, stable eigenvalues of $A,$ if there are any, can be ignored by reducing the dimension of $A,$ because they only contribute to the transient response of the consensus trajectories \cite{Huang2012TAC1}.

\smallskip

Consider the distributed control protocol for each agent $i$ as \cite{li2014cooperative}-\cite{lewis2013cooperative}
\begin{equation}
u_{i}=cK\eta_i\,\,\,\,\,\,\,\forall\,\, i \in \mathcal{N}
\label{eq5}
\end{equation}
where 
\begin{equation}
\eta _i = \sum\limits_{j =1}^N {{a_{ij}}({x_j} - {x_i})} 
\label{eq7}
\end{equation}
represents the local neighborhood tracking error for the agent $i$ with $a_{ij}$ as the $(i,j)$-th entry of the graph adjacency matrix $\mathcal{A}$. Moreover, $c$ and $K\in {R^{m \times n}}$ denote, respectively, scalar coupling gain and feedback control gain matrix. The control design \eqref{eq5}-\eqref{eq7} is distributed in the sense that each agent seeks to make the difference between its state and those of its neighbors equal to zero using only relative state information of its neighbors provided by \eqref{eq7}.  

Several approaches are presented to design $c$ and $K$ locally to solve Problem 1 \cite{li2014cooperative}-\cite{LewisTAC2012}. To this end, the gains $K$ and $c$ are designed such that $A-c\lambda_{i}BK$ is Hurwitz for all $i=2,\ldots,N$ \cite{li2014cooperative}-\cite{LewisTAC2012}. Specifically, it is shown in \cite{li2014cooperative}, \cite{LewisTAC2012} that under Assumption 1, if $K$ is designed locally for each agent by solving an algebraic Riccati equation and $c>\frac{1}{2\lambda_{min}(\mathcal{L})}$, then $A-c\lambda_{i}BK$ is Hurwitz and Problem 1 is solved. In the subsequent sections, we assume that  $c$ and $K$ are designed appropriately so that in the absence of attacks Problem 1 is solved. We then analyze the effect of attacks and propose mitigation approaches. 

\smallskip

\noindent
\textbf{Remark 2.}
Note that the presented results subsume the leader-follower synchronization problem and the average consensus as special cases. For the leader-follower case, the leader is only root node in the graph and thus the desired trajectory is dictated by the leader, whereas for the  average consensus case, the graph is assumed to be balanced and  $A=0$ and $B=I_m$.


\section{Attack Modelling and Analysis for DMASs}
In this section, attacks on agents are modelled and a complete attack analysis is provided.
\subsection{Attack Modelling}
In this subsection, attacks on DMASs are modelled. Attacks on actuators of agent $i$ can be modelled as 
\begin{equation}\label{eq17a}
u_i^c = {u_i} + {\beta _i}u_i^d
\end{equation}
where ${u_i}$, $u_i^d$ and $u_i^c$ denote, respectively, the nominal value of the control protocol for agent $i$ in \eqref{eq1}, the disrupted signal directly injected into actuators of agent $i$, and the corrupted control protocol of agent $i$. If agent $i$ is under actuator attack, then ${\beta _i} = 1$, otherwise ${\beta_i} = 0$. Similarly, one can model attacks on sensors of agent $i$ as  
\begin{equation}\label{eq16}
x_i^c = {x_i} + {\alpha _i}x_i^d
\end{equation}
where ${x_i}$, $x_i^d$ and $x_i^c$ denote, respectively, the nominal value of the state of agent $i$ in \eqref{eq7}, the disrupted signal directly injected into sensors of agent $i$, and the corrupted state of agent $i$. If agent $i$ is under sensor attack, then ${\alpha _i} = 1$, otherwise ${\alpha _i} = 0$. Using the corrupted state \eqref{eq16} in the controller \eqref{eq5}-\eqref{eq7} with the corrupted control input \eqref{eq17a} in \eqref{eq1}, the system dynamics under attack becomes 
\begin{equation}\label{eq17}
{\dot x _i} = A{x _i} + B{u _i} + B{f_i}
\end{equation}
where ${f_i}$ denotes the \textit{overall attack} affecting the agent $i$ which can be written as
\begin{equation}\label{eq18}
{f_i} = {\beta _i}u_i^d+cK\sum\limits_{j \in {N_i}} {{a_{ij}}\left( {{\alpha _j}x_j^d} -{{\alpha _i}x_i^d}\right)}
\end{equation}
with $x _j^d$ as the disruption in the received state of the $j^{th}$ neighbor due to injected attack signal either into sensors or actuators of agent $j$ or into the incoming communication link from agent $j$ to agent $i$. \par
The following definition categorizes all attacks into two categories. The first type of attack exploits the knowledge of the system dynamics $A$ and use it in the design of its attack signal. That is, for the first type of attack for $f_i$ in \eqref{eq18}, one has 
\begin{equation}\label{eq21}
\dot f_i = {\Psi _{}}f_i
\end{equation}
where $\Psi  \in {\mathbb{R}^{m \times m}}$  depends on the knowledge of the system dynamics $A$ as discussed in Definition $3$. On the other hand, for the second type of attack, the attacker has no knowledge of the system dynamics $A$ and this can cover all other attacks that are not in the form of \eqref{eq21}.\par Define 
\begin{equation}\label{eq22}
\left\{ {\begin{array}{*{20}{c}}
  {{E_\Psi } = \{ {\lambda _1}(\Psi ), \ldots ,{\lambda _m}(\Psi )\} } \\ 
  {{E_A} = \{ {\lambda _1}(A), \ldots ,{\lambda _n}(A)\} } 
\end{array}} \right.
\end{equation}
where ${\lambda _i}(\Psi )$ $\forall i = 1, \ldots ,m$ and ${\lambda _i}(A)$ $\forall i = 1, \ldots ,N$ are, respectively, the set of eigenvalues of the attack signal generator dynamics matrix $\Psi $ and the system dynamics matrix $A$. Define a set of eigenvalues of the system dynamics matrix $A$ which lie on imaginary axis as ${{E_{A}^{m}} = \{ {\lambda _1}(A), \ldots ,{\lambda _l}(A)\} }$ where $E_{A}^{m} \subseteq {E_A}$.

\smallskip

\noindent
\textbf{Definition 3 (IMP-based and non-IMP-based Attacks).} If the attack signal $f_i$ in \eqref{eq17} is generated by \eqref{eq21}, then the attack signal is called the internal model principle (IMP)-based attack, if ${E_\Psi } \subseteq E_{A}$. Otherwise, i.e., ${E_\Psi } \not\subset E_{A}$ or if the attacker has no dynamics (e.g. a random signal), it is called a non-IMP based attack.\hfill $\square$

\smallskip

\noindent
\textbf{Remark 3.} {Note that \textit{we do not limit attacks to the IMP-based attacks} given by \eqref{eq21}.} Attacks are placed into two classes in Definition 3 based on their impact on the system performance, as to be shown in the subsequent sections. The \textit{non-IMP based attacks cover a broad range of attacks.}

\smallskip

\noindent
\textbf{Definition 4 (Compromised and Intact Agent).} We call an agent that is directly under attack as a compromised agent. An agent is called intact if it is not compromised. We denote the set of intact agents as $\mathcal{N}_{Int}$,  $i.e., \mathcal{N}_{Int}= \mathcal{N} \backslash \mathcal{N}_{Comp}$ where $\mathcal{N}_{Comp}$ denotes the set of compromised agents.\par\hfill $\square$

\smallskip

Using \eqref{eq5}-\eqref{eq7}, the global form of control input, i.e., $u  = {[ {u_1^T, \ldots ,u_N^T} ]^T}$ can be written as
\begin{equation}\label{eq24a}
u=(- c\mathcal{L} \otimes {K})x 
\end{equation}
where $\mathcal{L}$ denotes the graph Laplacian matrix. \par
By using \eqref{eq24a} in \eqref{eq17}, the global dynamics of agents under attack becomes
\begin{equation}\label{eq23}
\dot x(t) = \left( {{I_N} \otimes A} \right)x(t)  +\left( {{I_N} \otimes B} \right)u^{c}(t)
\end{equation}
where
\begin{equation}\label{eq24}
u^{c}(t)=u(t)+f(t) \buildrel \Delta \over = (- c\mathcal{L} \otimes {K})x(t)  + f(t) 
\end{equation}
with $f(t) = {[ {f_1^T(t), \ldots ,f_N^T(t)} ]^T}$ and $x(t)= {[ {x_1^T(t), \ldots ,x_N^T(t)} ]^T}$ denote, respectively, the overall vector of attacks on agents and the global vector of the states of agents.\par

If agents are not under attack, i.e., $f(t)=0$, then, the control input \eqref{eq24a} eventually compensates for the difference between the agents' initial conditions and becomes zero once they reach an agreement. That is, in the absence of attack, $u^c(t)=u(t)$ goes to zero (i.e., $u^c(t) \to 0$), and, on consensus, the global dynamics of agents become
\begin{equation}\label{eq24b}
\dot x_{ss}(t) = \left( {{I_N} \otimes A} \right)x_{ss}(t) 
\end{equation}
where $x_{ss}={\mathop {\lim }\limits_{t \to \infty } }x(t)$ is called the global steady state of agents. Throughout the paper, regardless of whether agents are under attack or not, we say that agents reach a steady state and their steady state is generated by \eqref{eq24b} if $u^c(t) \to 0$. Otherwise, if $u^c(t) \not \to 0$, we say agents never reach a steady state, and thus \eqref{eq24b} does not hold true.\par

\smallskip

\noindent
{\textbf{Remark 4.} In the presence of attack, whether agents reach a steady state or not, i.e., whether $u^c(t) \to 0$ or $u^c(t) \not \to 0$, plays an important role in the attack analysis and mitigation to follow. Reaching a steady state is necessary for agents to achieve consensus. However, we show that even if agents reach a steady state, they will not achieve consensus if the system is under attack. More specifically, we show that under a non-IMP based attack, agents do not reach a steady state and their local neighborhood tracking errors also do not converge to zero. For an IMP-based attack, the attacker can either 1) make all agents reach a steady state, but agents are still far from synchronization or consensus, or 2) destabilize the entire network by assuring that agents do not reach a steady state.}\par


\subsection{Attack Analysis}
In this subsection, a graph theoretic-based approach is utilized to analyze the effect of attacks on DMASs. To this end, the following notation and lemmas are used.

Let the graph Laplacian matrix $\mathcal{L}$ be partitioned as
\begin{equation}\label{eq25}
\mathcal{L} = \left[ {\begin{array}{*{20}{c}}
  {{\mathcal{L}_{r \times r}}}&{{0_{r \times nr}}} \\ 
  {{\mathcal{L}_{nr \times r}}}&{{\mathcal{L}_{nr \times nr}}} 
\end{array}} \right]{\mkern 1mu} 
\end{equation}
where $r$ and $nr$ in \eqref{eq25} denote, respectively, the number of root nodes and non-root nodes. {Moreover, ${\mathcal{L}_{r \times r}}$ and ${\mathcal{L}_{nr \times nr}}$ are, respectively, the sub-graph matrices corresponding to the sub-graphs of root nodes and non-root nodes.}
\smallskip

\noindent
\textbf{Lemma 2.} Consider the partitioned graph Laplacian matrix \eqref{eq25}. Then, ${\mathcal{L}_{r \times r}}$ is a singular M-matrix and ${\mathcal{L}_{nr \times nr}}$ is a non-singular M-matrix.

\begin{proof}
{We first prove that the subgraph of root nodes is strongly connected. According to the definition of a root node, there always exists a directed path from a root node to all other nodes of the graph ${\mathcal{G}}$, including other root nodes. Therefore, in the graph ${\mathcal{G}}$, there always exists a path from each root node to all other root nodes. We now show that removing non-root nodes from the graph ${\mathcal{G}}$ does not affect the connectivity of the subgraph comprised of only root nodes. In the graph ${\mathcal{G}}$, if a non-root node is not an incoming neighbor of a root node, then its removal does not harm the connectivity of the subgraph of the root nodes. Suppose that removing a non-root node affects the connectivity  of the subgraph of root nodes. This requires the non-root node to be an incoming neighbor of a root  node. However, this makes the removed node a root node, as it can now access all other nodes through the root node it is connected to. Hence, this argument shows that the subgraph of root nodes is always strongly connected. Then, based on Lemma $1$, ${\mathcal{L}_{r \times r}}$ has zero as one of its eigenvalues, which implies that ${\mathcal{L}_{r \times r}}$ is a singular M-matrix according to Definition $1$. On the other hand, from \eqref{eq25}, since $\mathcal{L}$ is a lower triangular matrix, the eigenvalues of $\mathcal{L}$ are the union of the eigenvalues of ${\mathcal{L}_{r \times r}}$ and ${\mathcal{L}_{nr \times nr}}$. Moreover, as stated in Lemma 1, $\mathcal{L}$ has a simple zero eigenvalue and, as shown above, zero is the eigenvalue of ${\mathcal{L}_{r \times r}}$. Therefore, all eigenvalues of ${\mathcal{L}_{nr \times nr}}$ have positive real parts only, and thus based on Definition $2$, ${\mathcal{L}_{nr \times nr}}$ is a non-singular M-matrix.}
\end{proof}



In the following Lemmas 3-4 and Theorem 1, we now provide the conditions under which the agents can reach a steady state.

\smallskip
\noindent
\textbf{Lemma 3.}  {Consider the global dynamics of DMAS \eqref{eq23} under attack. Let the attack signal $f(t)$ be a non-IMP based attack. Then, agents never reach a steady state, i.e., $u^c(t) \not \to 0$.}
\begin{proof}
We prove this result by contradiction. Assume that the attack signal $f(t)$ is a non-IMP based attack, i.e., ${E_\Psi }\not  \subset {E_A}$, but $u^c(t) \to 0$ in \eqref{eq23}, which implies ${\dot {x} _i \to A{x} _i}$ for all $i \in \mathcal{N}$. Using the modal decomposition, one has 
\begin{equation}\label{eq27}
{x _i}(t) \to \sum\limits_{j = 1}^n {({r_j}{x _i}(0)){e^{{\lambda _j}(A)t}}{m_j}}
\end{equation}
where ${r_j}$ and ${m_j}$ denote, respectively, the left and right eigenvectors associated with the eigenvalue ${\lambda _j}(A)$. On the other hand, based on \eqref{eq24} $u^c(t) \to 0$ implies $f(t) \to (c\mathcal{L} \otimes {K})x(t)$ or equivalently
\begin{equation}\label{eq28}
{f_i(t)}\to \sum\limits_{j \in {N_i}} {{a_{ij}}({x_j(t)} - {x_i(t)})} 
\end{equation}
for all $i \in \mathcal{N}$. As  shown in \eqref{eq27}, the right-hand side of \eqref{eq28} is generated by the natural modes of the system dynamics whereas the left-hand side is generated by the natural modes of the attack signal generator dynamics in \eqref{eq21}. By the prior assumption, ${E_\Psi }\not  \subset {E_A}$, the attacker's natural modes are different from those of the system dynamics. Therefore, \eqref{eq28} cannot be satisfied which contradicts the assumption. This completes the proof.
\end{proof}



Equation \eqref{eq28} in Lemma 3 also shows that for non-IMP based attacks, the local neighborhood tracking error is nonzero for a compromised agent. The following results show that under IMP-based attack, either agents' state diverge, or they reach a steady state while their local neighborhood tracking errors converge to zero, despite attack. The following lemma is needed in Theorem 1, which gives conditions under which agents reach a steady state under IMP-based attack. Then, Theorem $2$ shows that under what conditions an IMP-based attack makes the entire network of agents unstable. 

Define 
\begin{equation}\label{eq30b}
\left\{ {\begin{array}{*{20}{c}}
  S_A(t)=[e^{\lambda_{A_1}t},\ldots,e^{\lambda_{A_n}t}] \\ 
 S_{\psi}(t)=[e^{\lambda_{{\Psi}_1}t},\ldots,e^{\lambda_{{\Psi}_n}t}]
\end{array}} \right.
\end{equation}
where $e^{\lambda_{A_i}t}$ $\forall i = 1, \ldots ,n$ and $e^{\lambda_{{\Psi}_i}t}$ $\forall i = 1, \ldots ,m$ are, respectively, the set of natural modes of agent dynamics $A$ in \eqref{eq1} and the attacker dynamics $\Psi$ in \eqref{eq21}. 

\smallskip
\noindent
\textbf{Lemma 4.}  Consider the global dynamics of DMAS  \eqref{eq23} under attack on non-root nodes. Then, for an IMP-based attack, agents reach a steady state, i.e., $u^c(t)  \to 0$.

\begin{proof}
According to \eqref{eq24b}, in steady state, one has ${\dot x_{ss}}(t) \to \left( {{I_N} \otimes A} \right){x _{ss}(t)}$ since $u^c(t) \to 0$. This implies that ${x_{ss}}(t) \in \mathrm{span}(S_{A})$ where $S_{A}$ is defined in \eqref{eq30b}. On the other hand, if agents reach a steady state, then based on \eqref{eq24}, one has  
\begin{equation}\label{eq30a}
(c\mathcal{L} \otimes {K})x_{ss}(t)=f(t)
\end{equation}
\par
Define the global steady state  vector ${x_{ss}(t)} = {[\bar x_{rs}^T,\bar x_{nrs}^T]^T}$, where ${\bar x_{rs}}$ and ${\bar x_{nrs}}$ are, respectively, the global steady states of root nodes and non-root nodes. Since attack is only on non-root nodes, $f(t)$ can be written as $f(t) = {[{0_r},\bar f_{nr}^T]^T}$, where ${\bar f_{nr}} = {[f_{r + 1}^T, \ldots ,f_N^T]^T}$ represents the attack vector on non-root nodes.\par
Then, using \eqref{eq25} and \eqref{eq30a}, one has
\begin{equation}\label{eq30}
\left\{ {\begin{array}{*{20}{c}}
  {({c\mathcal{L}_{r \times r}} \otimes {K}){{\bar x }_{rs}} = 0} \\ 
  {({c\mathcal{L}_{nr \times r}} \otimes {K}){{\bar x}_{rs}} + ({c\mathcal{L}_{nr \times nr}} \otimes {K}){{\bar x }_{nrs}} = {{\bar f}_{nr}}} 
\end{array}} \right.
\end{equation}
As stated in Lemma 2, ${\mathcal{L}_{r \times r}}$ is a singular M-matrix with zero as an eigenvalue and ${{\mathbf{1}}_r}$ is its corresponding right eigenvector and, thus, the solution to the first equation of \eqref{eq30} becomes ${\bar x _{rs}} = {c_1}{{\mathbf{1}}_r}$ for some positive scalar ${c_1}$. Using ${\bar x _{rs}} = {c_1}{{\mathbf{1}}_r}$ in the second equation of \eqref{eq30}, the global steady states of non-root nodes becomes
\begin{equation}\label{eq31}
{\bar x_{nrs}} = {({c\mathcal{L}_{nr \times nr}} \otimes {K})^{ - 1}}\left[ { - ({c\mathcal{L}_{nr \times r}} \otimes {K}){c_1}{{\mathbf{1}}_r} + {{\bar f}_{nr}}} \right] 
\end{equation}
Equation \eqref{eq31} shows that the steady states of non-root nodes are affected by the attack signal $f(t)$. If ${E_\Psi }\not  \subset {E_A}$, it results in ${\bar x_{nrs}} \in \mathrm{span}(S_A, S_{\Psi})$ where $S_{A}$ and $S_{\Psi}$ are defined in \eqref{eq30b} which contradicts ${x_{ss}}(t) \in \mathrm{span}(S_{A})$. Therefore, condition ${E_\Psi } \subset {E_A}$ is necessary to conclude that for any $f = {[{0_r},\bar f_{nr}^T]^T}$, there exists a steady state solution ${x_{ss}(t)}$, i.e., $u^c(t) \to 0$ holds true. This completes the proof.
\end{proof}

{The following theorem provides necessary and  sufficient conditions for IMP-based attacks to assure $u^c(t)  \to 0$.}
\begin{theorem}
{Consider the global dynamics of DMAS  \eqref{eq23} with the control protocol \eqref{eq24}, where the attack signal $f(t)$ is generated based on an IMP-based attack. Then, agents reach a steady state, i.e., $u^c(t) \to 0$ if and only if the attack signals satisfy
\begin{equation}\label{eq29}
	\sum\limits_{k = 1}^N {{p_k}} {f_k} = 0
\end{equation}
where ${p_k}$ are the nonzero elements of the left eigenvector of the graph Laplacian matrix $\mathcal{L}$ associated with its zero eigenvalue.}
\end{theorem}
\begin{proof}
 It was shown in the Lemma 4 that for the IMP-based attack on non-root nodes, agents reach a steady state, i.e., $u^c(t)  \to 0$. Therefore, whether agents reach a steady state or not depends solely upon attacks on root nodes.
 Let $f(t)=[\bar f_r, \bar f_{nr}]$ where $\bar f_r$ represents the vector of attacks for root nodes given by ${\bar f_r} = {[f_1^T, \ldots ,f_r^T]^T}$. Now, we first prove the necessary condition for root nodes. If $u^c(t)  \to 0$, then, using \eqref{eq25} and \eqref{eq30a} , there exists a nonzero vector ${\bar x_{rs}}$ for root nodes such that
\begin{equation}\label{eq32}
({c\mathcal{L}_{r \times r}} \otimes {K}){\bar x_{rs}} = {\bar f_r}
\end{equation}
where ${\bar x_{rs}}$ can be considered as the global steady state of the root nodes. {Moreover, based on Lemma 3, \eqref{eq32} does not hold, if ${E_\Psi } \not \subset {E_A}$ which implies that \eqref{eq32} is true only for ${E_\Psi } \subseteq {E_A}$}. As stated in Lemma 2, ${\mathcal{L}_{r \times r}}$ is a strongly connected graph of root nodes and, therefore, it is a singular M-matrix. Let  ${\bar w^T} = [{p_1}, \ldots ,{p_r}]$ be the left eigenvector associated with the zero eigenvalue of ${\mathcal{L}_{r \times r}}$. Now, pre-multiplying both sides of \eqref{eq32} by ${\bar w^T}$ and using the fact that ${\bar w^T}{\mathcal{L}_{r \times r}} = 0$ yield 
\begin{equation}\label{eq33}
{\bar w^T}({c\mathcal{L}_{r \times r}} \otimes {K}){\bar x _{rs}} = {\bar w^T}{\bar f_r} = 0
\end{equation}
This states that IMP-based attacks on root nodes have to satisfy $\sum\limits_{k = 1}^N {{p_k}{f_k} = 0} $ to ensure agents reach a steady state, i.e., $u^c(t)  \to 0$. Note that $p_k=0$ for $k=r+1,\ldots, N$, i.e., the elements of the left eigenvector of the graph Laplacian matrix $\mathcal{L}$, corresponding to its zero eigenvalue, are zero for non-root nodes  \cite{li2014cooperative}-\cite{lewis2013cooperative}. This proves the necessity part. 

Now, we prove the sufficient part by contradiction for root nodes. Assume agents reach a steady state, i.e., $u^c(t)  \to 0$, but $\sum\limits_{k = 1}^N {{p_k}{f_k} \ne 0}$. Note that, agents reach a steady state implies that there exists a nonzero vector ${\bar x _{rs}}$ such that \eqref{eq32} holds. Using \eqref{eq33} and $\sum\limits_{k = 1}^N {{p_k}{f_k} \ne 0} $, one can conclude that ${\bar w^T}({c\mathcal{L}_{r \times r}} \otimes {K}){\bar x_{rs}} \ne 0$. This can happen only when ${\mathcal{L}_{r \times r}}$ does not have any zero eigenvalue, which violates the fact in Lemma $2$ that ${\mathcal{L}_{r \times r}}$ is a strongly connected graph. Therefore, ${\bar w^T}({c\mathcal{L}_{r \times r}} \otimes {K}){\bar x_{rs}} = 0$ which results in $\sum\limits_{k = 1}^N {{p_k}{f_k} = 0} $ and contradicts the assumption made. This completes the proof.
\end{proof}

\noindent
\begin{theorem}
Consider the global dynamics of DMAS  \eqref{eq23} with the control protocol \eqref{eq24} under IMP-based attack. If \eqref{eq29} is not satisfied and ${E_\Psi } \cap {E_A^m} \ne \emptyset$, then the dynamics of agents become unstable.
\end{theorem}
\begin{proof}
Since {it is assumed that the condition} in \eqref{eq29} is not satisfied, then based on Theorem $1$, $u^c(t) \not  \to 0$ even under IMP-based attack.  Thus, the attack signal $f(t)$ does not vanish over time and eventually acts as an input to the system in \eqref{eq23}. Assume that there exists at least one common marginal eigenvalue between the system dynamics matrix $A$ in \eqref{eq1} and the attacker dynamics matrix $\Psi$ in \eqref{eq21}, i.e., ${E_\Psi } \cap {E_A^m} \ne \emptyset $. {Then, the multiplicity of at least one marginally stable pole becomes greater than 1.} Therefore, the attacker destabilizes the state of the agent in \eqref{eq23}. Moreover, since \eqref{eq29} is not satisfied, then the attack is on root nodes, and root nodes have a path to all other nodes in the network, the state of the all agents become unstable. This completes the proof.
\end{proof}


Theorem 3 below now shows that despite IMP-based attacks, if $u^c(t) \to 0$, the local neighborhood tracking error \eqref{eq7} converges to zero for intact agents that have a path to the compromised agent, while they do not synchronize.

\begin{theorem}
Consider the global dynamics of DMAS  \eqref{eq23} under attack $f(t)$. Then, the local neighborhood tracking error \eqref{eq7} converges to zero for all intact agents if $u^c(t) \to 0$. Moreover, intact agents that are reachable  from the compromised agents do not converge to the desired consensus trajectory.
\end{theorem}
\begin{proof}
In the presence of attacks, the global dynamics of the DMAS \eqref{eq23} with \eqref{eq24} can be written as
\begin{equation}\label{eq34}
\dot x(t) = \left( {{I_N} \otimes A} \right)x(t)  +\left( {{I_N} \otimes B})((- c\mathcal{L} \otimes {K})x(t)  + f(t) \right)
\end{equation}
where $x(t)  = {\left[ {x _1^T(t), \ldots ,x_N^T(t)} \right]^T}$ is the global vector of the state of agents and $f(t) = {\left[ {f_1^T(t), \ldots ,f_N^T(t)} \right]^T}$ denotes the global vector of attacks. As shown in \eqref{eq24b} that if $u^c(t) \to 0$, agents reach a steady state. That is, 
\begin{equation}\label{eq36}
{cK\eta _i} \to  - {f_i} \,\,\,\,\,\,\, \forall\,\, i \in \mathcal{N}
\end{equation}
where $\eta_i$ denotes the local neighborhood tracking error of agent $i$ defined in \eqref{eq7}. For the intact agent, by definition one has ${f_i} = 0$, and thus \eqref{eq36} implies that the local neighborhood tracking error \eqref{eq7} converges to zero. Now, we show that intact agents which are reachable from the compromised agent do not synchronize to the desired consensus behavior. To do this, let agent $j$ be under attack. Assuming that all intact agents synchronize, one has ${x_k} = {x_i}$  $\forall i,k \in \mathcal{N}  - \{ j\} $. Now, consider the intact agent $i$ as an immediate neighbor of the compromised agent $j$. Then using \eqref{eq24}, if $u^c(t) \to 0$, for intact agent $i$, i.e., ${f_i} = 0$, one has
\begin{equation}\label{eq37}
{\sum\limits_{k \in {N_i}-\left\{ j \right\}} {{a_{ij}}({x_k} - {x_i}) + ({x_j} - {x_i}) \to 0}}
\end{equation}
where ${x_k}$ denotes the state of the all intact neighbors of agent $i$. On the other hand, \eqref{eq17} shows that the state of the compromised agent $j$, i.e., ${x_j}$, is deviated from the desired consensus value with a value proportional to ${f_j}$. Therefore, \eqref{eq37} results in deviating the state of the immediate neighbor of the compromised agent $j$ from the desired consensus behavior, which contradicts the assumption. Consequently, intact agents that have a path to the compromised agent do not reach consensus, while their local neighborhood tracking error is zero. This completes the proof.
\end{proof}

\noindent
\textbf{Remark 5.} The effects of an attacker on a network of agents depend upon the dynamics of the attack signal. As stated in Theorem 2, to destabilize the entire network, the attack signal requires access to at least one common marginal eigenvalue with the system dynamics. To this end, an attacker can exploit the security of the network by eavesdropping and monitoring the transmitted data to identify at least one of the marginal eigenvalues of the agent dynamics, and then launch a signal with the same frequency to a root node to make the agents state go to infinity.

\smallskip

\noindent
\textbf{Remark 6.} Although, for the sake of simplicity, we consider DMASs with identical dynamics, the presented result can be extended to heterogeneous MASs.  This is briefly discussed in the following formulation. The dynamics for a linear heterogeneous MASs is given by 
\begin{equation}\label{eq1a}
\left\{ {\begin{array}{*{20}{c}}
  {{{\dot x}_i}(t) = {A_i}{x_i}(t) + {B_i}{u_i}(t)} \\ 
  {{y_i}(t) = {C_i}{x_i}(t)} 
\end{array}} \right.\begin{array}{*{20}{c}}
  {}&{\forall\,\, i \in \mathcal{N}} 
\end{array}
\end{equation}
where ${x_i} \in {\mathbb{R}^{{n_i}}},\,\,{u_i} \in {\mathbb{R}^{{m_i}}}$ and ${y_i} \in {\mathbb{R}^p}$ denote, respectively, the state, the control input and the output of agent $i$. The matrices ${A_i} \in {\mathbb{R}^{{n_i} \times {n_i}}},{B_i} \in {\mathbb{R}^{{n_i} \times {m_i}}}$ and ${C_i} \in {\mathbb{R}^{p \times {n_i}}}$ are, respectively, the drift dynamics, the input matrix and the output matrix. \par
For heterogeneous MASs, the consensus trajectory is usually generated by a virtual exosystem dynamics given by \cite{Allgower2011}-\cite{LunzeTAC2012}
\begin{equation}\label{eq2a}
\left\{ {\begin{array}{*{20}{c}}
  {{{\dot x}_c}(t) = {S_{}}{x_c}(t)} \\ 
  {{y_c}(t) = {R_{}}{x_c}(t)} 
\end{array}} \right.
\end{equation}
where ${x_c} \in {\mathbb{R}^q}$ and ${y_c} \in {\mathbb{R}^p}$ are, respectively, the state and output of the desired consensus trajectory. For heterogeneous MASs the distributed control protocol $u_i$ in \eqref{eq1a} is designed such that all agents synchronize to the output of virtual exosystem trajectory \cite{Allgower2011}-\cite{LunzeTAC2012}.
The attacker can design IMP-based attacks by exploiting the knowledge of consensus dynamics $S$ in \eqref{eq2a}, instead of agents' dynamics and all the analysis results presented in Section IV are valid for the heterogeneous MASs. In this case, to launch an IMP-based attack, the attacker should satisfy ${E_\Psi } \subseteq {E_S}$ where ${{E_S} = \{ {\lambda _1}(S), \ldots ,{\lambda _q}(S)\} } $ with ${\lambda _i}(S)$ $\forall i = 1, \ldots ,q$ as the set of eigenvalues of the virtual exosystem drift dynamics matrix $S$.


Up to now, the presented analysis has been under the assumption that the communication is noise free. We now briefly discuss what changes if the communication noise is present, and propose attack detection and mitigation in the presence of communication noise. In the presence of Gaussian distributed communication noise, the local neighborhood tracking error in \eqref{eq7} becomes
\begin{equation}\label{eq12}
{\bar{\eta} _i} = {\eta} _i + {\omega _i},
\end{equation}
where ${\omega _i} \sim \mathcal{N}(0,{\Sigma _{{\omega _i}}})$ denotes the aggregate Gaussian noise affecting the incoming information to agent $i$ and is given as
\vspace{-0.2cm}
\begin{equation}\label{eq13}
{\omega _i} = \sum\limits_{j \in {N_i}} {{a_{ij}}{\omega _{ij}}},
\vspace{-0.2cm}
\end{equation}
with ${\omega _{ij}}$ the incoming communication noise from agent $j$ to agent $i$. In such situations, the DMAS consensus problem defined in Problem $1$ changes to the mean square consensus problem.
In the presence of Gaussian noise, based on \eqref{eq12}, the control protocol in  \eqref{eq5}-\eqref{eq7} becomes
\vspace{-0.2cm}
\begin{equation}\label{eq45a}
u_{i}(t)=cK\sum\limits_{j =1}^N {{a_{ij}}({x_j}(t) - {x_i}(t))}+{\omega_i}\,\,\,\,\,\,\,\forall\,\, i \in \mathcal{N},
\vspace{-0.2cm}
\end{equation}
where ${\omega _i} \sim \mathcal{N}(0,{\Sigma _{{\omega _i}}})$ is defined in \eqref{eq13}. Based on mean square consensus, one has
\vspace{-0.2cm}
\begin{equation}\label{eq45b}
\begin{array}{*{20}{c}}
  {\mathop {\lim }\limits_{t \to \infty } \mathbb{E}{{\left[u_{i}(t)\right]}\to }0}&{}&{\forall\,\, i \in \mathcal{N}},
  \end{array}
\vspace{-0.0cm}
\end{equation}
and thus, based on \eqref{eq1}, the steady state of agents converge to a consensus trajectory in mean square sense and its global form in \eqref{eq24b} becomes 
\vspace{-0.2cm}
\begin{equation}\label{eq45c}
\dot x_{ss}^m= \left( {{I_N} \otimes A} \right)x_{ss}^m,
\vspace{-0.2cm}
\end{equation}
where $x_{ss}^m={{\lim }}_{t \to \infty }\mathbb{E}{[x(t)]}$ denotes the global steady state of agents in mean square sense. Then, following the same procedure as Lemmas 3-4 and Theorems 1-3, one can show that an IMP-based attack does not change the statistical properties of the local neighborhood tracking error, while a non-IMP based attack does. Moreover, the local neighborhood tracking error converges to zero in mean for an IMP-based attack, and it does not converges to zero in mean for a non-IMP based attack. 

In the next section, attack detection and mitigation mechanisms are proposed for both IMP-based and non-IMP based attacks. To this end, it is assumed that the communication network is noisy. 

\section{An Attack Detection Mechanism}
In this section, Kullback-Liebler (KL)-based attack detection and mitigation approaches are developed for both IMP-based and non-IMP-based attacks.

The KL divergence is a non-negative measure of the relative entropy between two probability distributions \cite{kullback1951information,basseville1993detection} which is defined as follows.
\smallskip

\noindent
\textbf{Definition 5 (KL divergence) \cite{kullback1951information,basseville1993detection}.} Let $X$ and $Z$ be two random sequences with probability density functions ${P_X}$ and ${P_Z}$, respectively. The KL divergence measure between ${P_X}$ and ${P_Z}$ is defined as
\begin{equation}
{D_{KL}}(X||Z) = \int {{P_X}(\theta )\log \left( {\frac{{{P_X}(\theta )}}{{{P_Z}(\theta )}}} \right)d\theta }
\label{m1}
\end{equation}
with the following properties \cite{kullback1951information}:
\begin{enumerate}
\item ${D_{KL}}({P_X}||{P_z}) \geqslant 0$
\item ${D_{KL}}({P_X}||{P_z}) = 0$ if and only if, ${P_X} = {P_z}$
\item ${D_{KL}}({P_X}||{P_z}) \ne {D_{KL}}({P_z}||{P_X})$
\end{enumerate}	 

In the following subsections, KL-divergence is used to detect IMP-based and non-IMP-based attacks on DMASs.
\vspace{-0.3cm}
\subsection{Attack detection for IMP-based attacks}
In this subsection, an attack detector is designed to identify IMP-based attacks. To this end,  two error sequences  ${\tau _i}$ and ${\varphi _i}$  are defined based on only local exchanged information for agent $i$ as 
\begin{equation}
{\tau _i} = \left\| {\sum\limits_{j \in {N_i}} {{a_{ij}d_{ij}}}} \right\|
\label{m2}
\end{equation}
and\\
\begin{equation}
{\varphi _i} = \sum\limits_{j \in {N_i}} {\left\| a_{ij}d_{ij} \right\|}
\label{m3}
\end{equation}
where the measured discrepancy $d_{ij}$ between agent $i's$ state and its neighbor $j's$ state under attack becomes
\begin{equation}
{d_{ij}} = {{x_j^c} - {x_i^c}} + {\omega _{ij}}\,\,\,\,\,\forall j \in {\mathcal{N}_i} 
\label{m3a}
\end{equation}
where ${\omega _{ij}} \sim \mathcal{N}(0,{\Sigma _{{\omega _{ij}}}})$ denotes the Gaussian incoming communication noise from agent $j$ to agent $i$. Moreover, ${x_i^c}$ is the measured state of agent $i$ under attack and ${x_j^c}$ is the possibly corrupted information it receives from its $j^{th}$ neighbor. If agent $i$ is not compromised, then ${x_i^c}={x_i}$, and, similarly, if agent $j$ is not compromised, then ${x_j^c}={x_j}$.  In fact, \eqref{m2} is the norm of the summation of the measured discrepancy of agent $i$ and all its neighbors, and \eqref{m3} is the summation of norms of those measured discrepancies. In the absence of attack, these two signals show the same behavior in the sense that their means converge to zero.  \par
In the presence of an IMP-based attack and in the absence of noise, based on Theorem 3, ${\tau _i}$ goes to zero for intact agents, despite attack. However, it is obvious that  ${\varphi _i}$ does not converge to zero in the presence of an attack. In the presence of noise, the statistical properties of ${\tau _i}$ converge to the statistical properties of the noise. In contrast, the statistical properties of ${\varphi _i}$ depend upon not only the statistical properties of the noise signal, but also of the attack signal. Therefore, the behavior of these two signals significantly diverges in the presence of attacks and can be captured by KL-divergence methods. {Note that one can measure ${\tau _i}$ and ${\varphi _i}$ based on the exchanged information among agents, which might be corrupted by the attack signal.} Existing KL-divergence methods are, nevertheless, developed for Gaussian signals. However, while the communication noise is assumed to be Gaussian, error sequences \eqref{m2} and \eqref{m3} are norms of some variable with Gaussian distributions, thus, they have univariate folded Gaussian distributions given by \cite{johnson2017process} ${\varphi _i} \sim \mathcal{FN}({\mu _{1i}},\sigma _{1i}^2)\,$ and ${\tau _i} \sim \mathcal{F}\mathcal{N}({\mu _{2i}},\sigma _{2i}^2)\,$. That is, 
\begin{align}
\begin{gathered}
  {P_{{\varphi _i}}}({q_i},{\mu _{1i}},{\sigma _{1i}}) = \frac{1}{{\sqrt {2\pi } \left| {{\sigma _{1i}}} \right|}}{e^{ - \frac{{{{({q_i} - {\mu _{1i}})}^2}}}{{2\sigma _{1i}^2}}}} + \frac{1}{{\sqrt {2\pi } \left| {{\sigma _{1i}}} \right|}}{e^{ - \frac{{{{({q_i} + {\mu _{1i}})}^2}}}{{2\sigma _{1i}^2}}}} \hfill \\
  {P_{{\tau _i}}}({q_i},{\mu _{2i}},{\sigma _{2i}}) = \frac{1}{{\sqrt {2\pi } \left| {{\sigma _{2i}}} \right|}}{e^{ - \frac{{{{({q_i} - {\mu _{2i}})}^2}}}{{2\sigma _{2i}^2}}}} + \frac{1}{{\sqrt {2\pi } \left| {{\sigma _{2i}}} \right|}}{e^{ - \frac{{{{({q_i} + {\mu _{2i}})}^2}}}{{2\sigma _{2i}^2}}}} \hfill \\ 
\end{gathered}
\label{m4}
\end{align}
where ${\mu _{1i}}$ and ${\sigma _{1i}}$ are the mean and variance of the error sequences ${\varphi _i}$ and ${\mu _{2i}}$ and ${\sigma _{2i}}$ are the mean and variance of the error sequences ${\tau _i}$. Using \eqref{m1}, the KL divergence in terms of the local error sequences ${\varphi _i}$ and ${\tau _i}$ can be defined as
\begin{equation}
{D_{KL}}({\varphi _i}||{\tau _i}) = \int {{P_{{\varphi _i}}}({q_i})\log \left( {\frac{{{P_{{\varphi _i}}}({q_i})}}{{{P_{{\tau _i}}}({q_i})}}} \right)} d{q_i} = {\mathbb{E}_1}\left( {\log \frac{{{P_{{\varphi _i}}}({q_i})}}{{{P_{{\tau _i}}}({q_i})}}} \right)
\label{m5}
\end{equation}
where ${\mathbb{E}_1}[.]$ represents the expectation value with respect to the distribution of the first sequence \cite{johnson2017process}.\par
A KL divergence formula for the folded Gaussian distributions is now developed in the following lemma.\par
\smallskip

\noindent
\textbf{Lemma 5.} Consider the error sequences  ${\tau _i}$ and ${\varphi _i}$ in \eqref{m2}-\eqref{m3} with folded Gaussian distributions ${P_{{\varphi _i}}}$ and ${P_{{\tau _i}}}$ in \eqref{m4}. Then, the KL divergence between error sequences ${\tau _i}$ and ${\varphi _i}$, i.e., ${D_{KL}}({\varphi _i}||{\tau _i})$, becomes
\begin{small}
\begin{align}
\begin{gathered}
  {D_{KL}}({\varphi _i}||{\tau _i}) \approx \frac{1}{2}\left( {\log \frac{{\sigma _{2i}^2}}{{\sigma _{1i}^2}} - 1 + (\sigma _{2i}^{ - 2}\sigma _{1i}^2)} \right) + \frac{1}{2}\sigma _{2i}^{ - 2}{({\mu _{2i}} - {\mu _{1i}})^2}+1 \hfill \\
  \frac{1}{2}{e^{\frac{{4\mu _{1i}^2}}{{\sigma _{1i}^2}}}}\left( {1 - {e^{\frac{{8\mu _{1i}^2}}{{\sigma _{1i}^2}}}}} \right) + {e^{ - \frac{{\mu _{1i}^2}}{{2\sigma _{1i}^2}}}}\left( {\frac{1}{2}\left( {{e^{\frac{{\rho _3^2}}{{2\sigma _{1i}^2}}}} + {e^{\frac{{\rho _4^2}}{{2\sigma _{1i}^2}}}}} \right) - \left( {{e^{\frac{{\rho _1^2}}{{2\sigma _{1i}^2}}}} + {e^{\frac{{\rho _2^2}}{{2\sigma _{1i}^2}}}}} \right)} \right)\hfill \\ 
\end{gathered}
\label{m6}
\end{align}
\end{small}
for some ${\rho _1} = ({\mu _{1i}} - 2{\mu _{2i}}\sigma _{1i}^2\sigma _{2i}^{ - 2}),$ ${\rho _2} = ({\mu _{1i}} + 2{\mu _{2i}}\sigma _{1i}^2\sigma _{2i}^{ - 2}),$  ${\rho _3} = ({\mu _{1i}} - 4{\mu _{2i}}\sigma _{1i}^2\sigma _{2i}^{ - 2})$ and ${\rho _4} = ({\mu _{1i}} + 4{\mu _{2i}}\sigma _{1i}^2\sigma _{2i}^{ - 2})$.
\begin{proof}
See Appendix A.
\end{proof}

In the following theorem, we show that the effect of IMP-based attacks can be captured using the KL divergence defined in \eqref{m6}.\par
\begin{theorem}
Consider the DMAS \eqref{eq1} along with the controller \eqref{eq24}, and under the IMP-based attacks. Assume that the communication noise sequences are i.i.d. Then, for a reachable intact agent $i$, 
\begin{equation}
\frac{1}{T}\int_k^{k + T - 1} {{D_{KL}}({\varphi _i}||{\tau _i})dk}  > {\gamma _i}
\label{m7}
\end{equation}
where ${\varphi _i}$ and ${\tau _i}$ are defined in \eqref{m2} and \eqref{m3}, respectively, and $T$ and ${\gamma _i}$ represent the window size and the predesigned threshold parameter.
\end{theorem}
\begin{proof}
According to Theorem 3, the local neighborhood tracking error goes to zero for intact agents in the presence of an IMP-based attack when there is no communication noise. In the presence of communication noise with Gaussian distribution, i.e., ${\omega _{ij}} \sim (0,{\Sigma _{{\omega _{ij}}}})$ and IMP-based attack, the expectation value of the local neighborhood tracking error for intact agent $i$ becomes 
\begin{equation}
\mathbb{E}[{\eta _i}] = \mathbb{E}[\sum\limits_{j \in {N_i}} {{a_{ij}}}d_{ij}] \to 0
\label{m8}
\end{equation}
where the measured discrepancy $d_{ij}$ is defined in \eqref{m3a}.
Using \eqref{m8}, one can write \eqref{m2} as
\begin{equation}
{\tau _i} = \left\| \sum\limits_{j \in {N_i}} {{a_{ij}}}d_{ij} \right\| \sim \mathcal{F}\mathcal{N}(0,\bar \upsilon _{\omega i}^2)\,
\label{m9}
\end{equation}
which represents a folded Gaussian distribution with mean zero and variance $\bar \upsilon _{\omega i}^2$. Note that the mean and variance of the distribution ${P_{{\tau _i}}}$  in \eqref{m4} become ${\mu _{2i}} = 0$ and $\sigma _{2i}^2 = \bar \upsilon _{\omega i}^2$.\par
Since noise signals are independent and identically distributed, from \eqref{m3}, one can infer that the folded Gaussian distribution ${P_{{\varphi _i}}}$ in \eqref{m4} has the following statistical properties
\begin{equation}
{\varphi _i} \sim \mathcal{F}\mathcal{N}({\mu _{f_i^d}},\bar \upsilon _{{\omega _i}}^2 + \hat \upsilon _{{\omega _i}}^2 + \bar \upsilon _{f_i^d}^2)\,
\label{m10}
\end{equation}
where ${\mu _{f_i^d}}$ and $\bar \upsilon _{{\omega _i}}^2 + \hat \upsilon _{{\omega _i}}^2 + \bar \upsilon _{f_i^d}^2$ represent the overall mean and covariance due to the communication noise and overall deviation from the desired behavior in intact neighbors reachable from the compromised agent.\par
In the absence of attack, the statistical properties corresponding to sequences ${\tau _i}$ and ${\varphi _i}$ become $\mathcal{F}\mathcal{N}(0,\bar \upsilon _{\omega i}^2)\,$ and $\mathcal{F}\mathcal{N}(0,\bar \upsilon _{\omega i}^2 + \hat \upsilon _{{\omega _i}}^2)\,$, respectively, and the corresponding KL divergence in \eqref{m6} becomes
\begin{equation}
D_{KL}^{wa}({\varphi _i}||{\tau _i}) \approx \frac{1}{2}\left( {\log \frac{{\bar \upsilon _{\omega i}^2}}{{\bar \upsilon _{{\omega _i}}^2 + \hat \upsilon _{{\omega _i}}^2}} + \bar \upsilon _{{\omega _i}}^{ - 2}\hat \upsilon _{{\omega _i}}^2))} \right)
\label{m11}
\end{equation}
where $\hat \upsilon _{{\omega _i}}^2$ represents additional variance in sequence ${\varphi _i}$, which depends on the communication noise. \par
Note that $\tau_i$ in \eqref{m2} is the norm of the summation of the measured discrepancy of agent $i$ and all its neighbors whereas $\varphi _i$ in  \eqref{m3} is the summation of norms of those measured discrepancies. Even in the absence of attack, they represent folded Gaussian distributions with zero means and different covariances due to application of norm on measured discrepancies.\par
Now, in the presence of IMP-based attacks, using the derived form of KL divergence for folded Gaussian distributions from Lemma 5, one can simplify \eqref{m6} using \eqref{m9}-\eqref{m10} as
\begin{align}
\begin{gathered}
 {D_{KL}}({\varphi _i}||{\tau _i}) \approx \frac{1}{2}\left( {\log \frac{{\bar \upsilon _{\omega i}^2}}{{\bar \upsilon _{{\omega _i}}^2 + \hat \upsilon _{{\omega _i}}^2 + \bar \upsilon _{f_i^d}^2}} + \bar \upsilon _{{\omega _i}}^{ - 2}(\bar \upsilon _{f_i^d}^2 + \hat \upsilon _{{\omega _i}}^2)} \right) \hfill \\
   + \frac{1}{2}\bar \upsilon _{{\omega _i}}^{ - 2}{({\mu _{f_i^d}})^2} + \frac{1}{2}{e^{\frac{{4{{({\mu _{f_i^d}})}^2}}}{{\bar \upsilon _{{\omega _i}}^2 + \hat \upsilon _{{\omega _i}}^2 + \bar \upsilon _{f_i^d}^2}}}}\left( {1 - {e^{\frac{{8{{({\mu _{f_i^d}})}^2}}}{{\bar \upsilon _{{\omega _i}}^2 + \hat \upsilon _{{\omega _i}}^2 + \bar \upsilon _{f_i^d}^2}}}}} \right) \hfill \\ 
\end{gathered} 
\label{m12}
\end{align}

Then, one can design the threshold parameter ${\gamma _i}$  such that
\begin{equation}
\frac{1}{T}\int_k^{k + T - 1} {{D_{KL}}({\varphi _i}||{\tau _i})dk}  > {\gamma _i}
\label{m13}
\end{equation}
where $T$ denotes the sliding window size. This completes the proof.  
\end{proof}

Based on Theorem 4, one can use the following conditions for attack detection.
\begin{align}
\left\{ \begin{gathered}
  \frac{1}{T}\int_k^{k + T - 1} {{D_{KL}}({\varphi _i}||{\tau _i})dk}  < {\gamma _i}\,\,\,:{H_0} \hfill \\
  \frac{1}{T}\int_k^{k + T - 1} {{D_{KL}}({\varphi _i}||{\tau _i})dk}  > {\gamma _i}\,\,\,\,\,:{H_1} \hfill \\ 
\end{gathered}  \right.
\label{m14}
\end{align}
where ${\gamma _i}$ denotes the designed threshold for detection, the null hypotheses ${H_0}$ represents the intact mode and ${H_1}$ denotes the compromised mode of an agent.
\vspace{-0.3cm}
\subsection{Attack detection for non-IMP-based attacks}
This subsection presents the design of a KL-based attack detector for non-IMP based attacks.\par
It was shown in Theorem 3 that the local neighborhood tracking error goes to zero if  agents are under IMP-based attacks. Therefore, for the case of non-IMP-based attacks, one can identify these types of attacks using the changes in the statistical properties of the local neighborhood tracking error. In the absence of attack, since the Gaussian noise, i.e., ${\omega _i} \sim \mathcal{N}(0,{\Sigma _{{\omega _i}}})$, is considered in the communication link, the local neighborhood tracking error ${\eta _i}$ in \eqref{eq12} has the following statistical properties
\begin{equation}
{\eta _i} \sim \mathcal{N}(0,{\Sigma _{{\omega _i}}})
\label{m15}
\end{equation}
and it represents the nominal behavior of the system.\par
In the presence of attacks, using \eqref{eq12}, the local neighborhood tracking error $\eta _i^a$ can be written as
\begin{equation}
\eta _i^a = \sum\limits_{j \in {N_i}} {{a_{ij}}}d_{ij}
\label{m16}
\end{equation}
where measured discrepancy under attacks ${d_{ij}}$ is defined \eqref{m3a}. From \eqref{m16}, one has
\begin{equation}
\eta _i^a \sim \mathcal{N}({\mu _{{f_i}}},{\Sigma _{{f_i}}} + {\Sigma _{{\omega _i}}})
\label{m17}
\end{equation}
where ${\mu _{{f_i}}}$ and ${\Sigma _{{f_i}}}$ are, respectively, mean and covariance of the overall deviation due to corrupted states under attacks as given in \eqref{m3a}.\par

Now, since both $\eta _i^a$ and ${\eta _i}$ have normal Gaussian distributions, the KL divergence in the terms of $\eta _i^a$ and ${\eta _i}$ as ${D_{KL}}(\eta _i^a||{\eta _i})$ can be written as \cite{KLBOOK}  
\begin{align}
\begin{gathered}
  {D_{KL}}(\eta _i^a||{\eta _i}) = \frac{1}{2}\left( {\log \frac{{\left| {{\Sigma _{{\eta _i}}}} \right|}}{{\left| {{\Sigma _{\eta _i^a}}} \right|}} - n + tr(\Sigma _{{\eta _i}}^{ - 1}{\Sigma _{\eta _i^a}})} \right) \hfill \\
  \,\,\,\,\,\,\,\,\,\,\,\,\,\,\,\,\,\,\,\,\,\,\,\,\,\,\,\,\, + \frac{1}{2}{({\mu _{{\eta _i}}} - {\mu _{\eta _i^a}})^T}\Sigma _{{\eta _i}}^{ - 1}({\mu _{{\eta _i}}} - {\mu _{\eta _i^a}}) \hfill \\ 
\end{gathered}
\label{m18}
\end{align}
where ${\mu _{{\eta _i}}}$ and ${\Sigma _{{\eta _i}}}$ denote the mean and covariance of ${\eta _i}$ and  ${\mu _{\eta _i^a}}$ and ${\Sigma _{\eta _i^a}}$ denote the mean and covariance of $\eta _i^a$. Moreover, $n$ denotes the dimension of the error sequence. Define the average of KL divergence over a window $T$ as
\begin{equation}
{\bar D_i} = \frac{1}{T}\int_k^{k + T - 1} {{D_{KL}}(\eta _i^a||{\eta _i})dk}
\label{m19}
\end{equation}
The following theorem says that the effect of non-IMP based attacks can be detected using the KL divergence between the two error sequences $\eta _i^a$ and ${\eta _i}$. 

\begin{theorem}
Consider the DMAS \eqref{eq1} along with the controller \eqref{eq24}. Then, 
\begin{enumerate}
\item in the absence of attack, ${\bar D_i}$ defined in \eqref{m19} tends to zero.
\item in the presence of a non-IMP-based attack, ${\bar D_i}$ defined in \eqref{m19} is greater than a predefined threshold ${\gamma _i}$. 
\end{enumerate}
\end{theorem}
\begin{proof}
In the absence of attacks, the statistical properties of sequences ${\eta _i}$ and $\eta _i^a$ are the same as in \eqref{m15}. Therefore, the KL divergence ${D_{KL}}(\eta _i^a||{\eta _i})$ in \eqref{m18} becomes zero and this makes ${\bar D_i}$ in \eqref{m19} zero. This completes the proof of part 1.\par
To prove Part 2, using \eqref{m15}-\eqref{m17} in \eqref{m18} and the fact that $(\Sigma _{{\omega _i}}^{ - 1}({\Sigma _{{f_i}}} + {\Sigma _{{\omega _i}}}) - n = tr(\Sigma _{{\omega _i}}^{ - 1}{\Sigma _{{f_i}}})$, one can write the KL divergence between $\eta _i^a$ and ${\eta _i}$ as
\begin{equation}
{D_{KL}}(\eta _i^a||{\eta _i}) = \frac{1}{2}(\log \frac{{\left| {{\Sigma _{{\omega _i}}}} \right|}}{{\left| {{\Sigma _{{f_i}}} + {\Sigma _{{\omega _i}}}} \right|}} + tr(\Sigma _{{\omega _i}}^{ - 1}{\Sigma _{{f_i}}}) + \mu _{{f_i}}^T\Sigma _{{\omega _i}}^{ - 1}{\mu _{{f_i}}})
\label{m22}
\end{equation}

Then, using \eqref{m19}, one has  
\begin{equation}
{\bar D_i} = \frac{1}{T}\int\limits_k^{k + T - 1} {\frac{1}{2}\log \frac{{\left| {{\Sigma _{{\omega _i}}}} \right|}}{{\left| {{\Sigma _{{f_i}}} + {\Sigma _{{\omega _i}}}} \right|}} + tr(\Sigma _{{\omega _i}}^{ - 1}{\Sigma _{{f_i}}}) + \mu _{{f_i}}^T\Sigma _{{\omega _i}}^{ - 1}{\mu _{{f_i}}})}  > {\gamma _i}
\label{m23}
\end{equation}
where $T$ and ${\gamma _i}$ denote the sliding window size and the predefined design threshold, respectively. This completes the proof.
\end{proof}

Based on Theorem 5, one can use the following conditions for attack detection:
\begin{align}
\left\{ \begin{gathered}
  {{\bar D}_i}\,\, < {\gamma _i}:{H_0} \hfill \\
  {{\bar D}_i}\,\, > {\gamma _i}\,\,:{H_1} \hfill \\ 
\end{gathered}  \right.
\label{m24}
\end{align}
where ${\gamma _i}$ denotes the designed threshold for detection, the null hypotheses ${H_0}$ represents the intact mode of the system and ${H_1}$ denotes the compromised mode of the system.\par
In the next section, Theorems 4 and 5 are employed to propose an attack mitigation approach which enables us to mitigate both IMP-based attacks and non-IMP-based  attacks.

\section{An Attack Mitigation Mechanism}
In this section, both IMP-based and non-IMP-based attacks are mitigated using the presented detection mechanisms in the previous section. 
\vspace{-0.3cm}
\subsection{Self-belief of agents about their outgoing information}
To determine the level of trustworthiness of each agent about its own information, a self-belief value is presented. If an agent detects an attack, it reduces its level of trustworthiness about its own understanding of the environment and communicates it with its neighbors to inform them about the significance of its outgoing information and thus slow down the attack propagation.\par
For the IMP-based attacks, using the ${D_{KL}}({\varphi _i}||{\tau _i})$ from Theorem 4, we define $c_i^1(t)$ as
\begin{equation}
c_i^1(t) = {\kappa _1}\int\limits_0^t {{e^{{\kappa _1}(\tau  - t)}}\chi _i^1(} \tau )d\tau
\label{m25}
\end{equation}
where  $0 \leqslant c_i^1(t) \leqslant 1$ with
\begin{equation}
\chi _i^1(t) = \frac{{{\Delta _i}}}{{{\Delta _i} + {D_{KL}}({\varphi _i}||{\tau _i})}}
\label{m26}
\end{equation}
where ${\Delta _i}$ represents the threshold to account for the channel fading and other uncertainties and ${\kappa _1} > 0$ denotes the discount factor. Equation \eqref{m25} can be implemented by the following differential equation
\begin{equation}
\nonumber
\begin{gathered}
  \dot c_i^1(t) + {\kappa _1}c_i^1(t) = {\kappa _1}\chi _i^1(t) \hfill \\ 
\end{gathered} 
\end{equation}

 According to Theorem 4, in the presence of IMP-based attacks, ${D_{KL}}({\varphi _i}||{\tau _i})$ increases, which makes $\chi _i^1(t)$ approach zero and consequently makes the value of $c_i^1(t)$ close to zero.  On the other hand, without an attack, ${D_{KL}}({\varphi _i}||{\tau _i})$ tends to zero, making $\chi _i^1(t)$ approach $1$ and, consequently, $c_i^1(t)$ becomes close to  $1$. The larger the value of $c_i^1(t)$ is, the more confident the agent is about  the trustworthiness  of its  broadcasted information.\par
Similarly, for the non-IMP-based attacks, using the ${D_{KL}}(\eta _i^a||{\eta _i})$ from Theorem 5, we define $c_i^2(t)$ as
\begin{equation}
c_i^2(t) = {\kappa _2}\int\limits_0^t {{e^{{\kappa _2}(\tau  - t)}}\chi _i^2(} \tau )d\tau 
\label{m27}
\end{equation}
where  $0 \leqslant c_i^2(t) \leqslant 1$ with
\begin{equation}
\chi _i^2(t) = \frac{{{\Delta _i}}}{{{\Delta _i} + {D_{KL}}(\eta _i^a||{\eta _i})}}
\label{m28}
\end{equation}
where ${\Delta _i}$ represents the threshold to account for the channel fading and other uncertainties, and ${\kappa _2} > 0$ denotes the discount factor. Expression \eqref{m27} can be generated by
\begin{equation}
\nonumber
\begin{gathered}
\dot c_i^2(t) + {\kappa _2}c_i^2(t) = {\kappa _2}\chi _i^2(t)
\end{gathered} 
\end{equation}

Using Theorem $6$ and the same argument as we employed for $c_i^1(t)$, one can show that $c_i^2(t)$ is close to $1$ in the absence of an attack, and close to zero in the presence of a non-IMP based attack.\par 
Then, using $c_i^1(t)$ and $c_i^2(t)$ defined in \eqref{m25} and \eqref{m27}, the self-belief of an agent $i$ for both IMP and non-IMP-based attacks is defined as
\begin{equation}
{\xi_i}(t) = \min \{ c_i^1(t),\,\,c_i^2(t)\} 
\label{m29}
\end{equation}

If an agent $i$ is under direct attack or receives corrupted information from its neighbors, then the self-belief of the agent $i$ tends to zero. In such a situation, it transmits the low self-belief value to its neighbor to put less weight on the information they receive from it and this prevents attack propagation in the distributed network.

\subsection{Trust of agents about their incoming information}
The trust value represents the level of confidence of an agent on its neighbors' information. If the self-belief value of an agent is low, it forms beliefs on its neighbors (either intact or compromised) and updates its trust value which depends on the beliefs on each of its neighbors using only local information. Therefore, agents identify the compromised neighbor and discard its information.\par
Using the KL divergence between exchanged information of agent $i$ and its neighbor, one can define ${\eta _{ij}}(t)$ as
\begin{equation}
{\eta _{ij}}(t) = {\kappa _3}\int\limits_0^t {{e^{{\kappa _3}(\tau  - t)}}{L_{ij}}(} \tau )d\tau \,\,\,\,
\label{m30}
\end{equation}
where  $0 \leqslant {\eta _{ij}}(t) \leqslant 1$ with 
\begin{equation}
{L_{ij}}(t) = 1 - \frac{{{\Lambda _1}}}{{{\Lambda _1} + {e^{\left( {\frac{{ - {\Lambda _2}}}{{{D_{KL}}({x_j}||{m_i})}}} \right)}}}}\begin{array}{*{20}{c}}
  {}&{\forall j \in {N_i}}&{} 
\end{array}
\label{m31}
\end{equation}
with ${m_i} = \sum\limits_{j \in {N_i}} {{x_j}}$ and ${\Lambda _1},{\Lambda _2} > 0$ represent the threshold to account for channel fading and other uncertainties, and ${\kappa _3} > 0$ denotes the discount factor. For the compromised neighbor, the KL divergence ${D_{KL}}({x_j}||{m_i})$ tends to zero, which makes ${L_{ij}}(t)$ close to zero. Consequently, this makes the  value of ${\eta _{ij}}(t)$ close to zero. On the other hand, if the incoming neighbor is not compromised, then ${D_{KL}}({x_j}||{m_i})$ increases and makes ${\eta _{ij}}(t)$ approach $1$.  Equation \eqref{m30} can be implemented using the following differential equation
\begin{equation}
\nonumber
\begin{gathered}
{\dot \eta _{ij}}(t) + {\kappa _3}{\eta _{ij}}(t) = {\kappa _3}{L_{ij}}(t)
\end{gathered} 
\end{equation}

Now, we define the trust value of an agent on its neighbors as
\begin{equation}
{\Omega _{ij}}(t) = \max ({\xi_i}(t),{\eta _{ij}}(t))
\label{m32}
\end{equation}
with $0 \leqslant {\Omega _{ij}}(t) \leqslant 1$.

In the absence of attacks, the state of agents converge to the consensus trajectory and the KL divergence ${D_{KL}}({x_j}||{m_i}),\,\,\forall j \in {N_i}$ tends to zero which results in ${\Omega _{ij}}(t)$ being $1$ $\forall j \in {N_i}$.  In the presence of attacks, ${\eta _{ij}}(t)$ corresponding to the compromised agents tends to zero. 
\vspace{-0.3cm}
\subsection{The mitigation mechanism using trust and self-belief values}
In this subsection, the trust and self-belief values are utilized to design the mitigation algorithm. To achieve resiliency, both self-belief and trust values are incorporated into the exchange information among agents. Consequently,  the resilient form of local neighborhood tracking error \eqref{eq12} is presented as
\begin{equation}
{\tilde{\eta} _i} = \sum\limits_{j \in {N_i}} {{\Omega _{ij}}(t){\xi_j}(t){a_{ij}}\left( {{x_j} - {x_i}} \right)}  + {\omega _i}
\label{m33}
\end{equation}
where ${\Omega _{ij}}(t)$ and ${\xi_j}(t)$ denote, respectively, the trust value and the self-belief of neighboring agents. Using \eqref{eq5} and \eqref{m33}, the resilient control protocol becomes
\begin{equation}
\tilde{u}_{i}=cK\tilde{\eta} _i,\,\,\,\,\,\,\,\forall\,\, i \in \mathcal{N}
\label{eq33a}
\end{equation}
According  to \eqref{m33}, the topology of the graph changes over time due to incorporation of the trust and the self-belief values of agents, thus we denote the time-varying graph as $\mathcal{G}(t) = (\mathcal{V},{\mkern 1mu} \mathcal{E}(t))$ with $\mathcal{E}(t) \subseteq \mathcal{V} \times \mathcal{V}$ representing the set of time-varying edges.\par


  

Now, based on following definitions and lemma, we formally present Theorem 6 to illustrate that the trust and self-belief based proposed resilient control protocol \eqref{eq33a} solves Problem 1 and all intact agents, i.e., $\mathcal{N}_{Int}= \mathcal{N} \backslash \mathcal{N}_{Comp}$ as defined in Definition 4 achieve the final desired consensus value for DMAS in \eqref{eq1}, despite attacks.
\smallbreak

\noindent
\textbf{Definition 6 (r-reachable set) \cite{Wupaper}.} Given a directed graph ${\mathcal{G}}$ and a nonempty subset ${\mathcal{V}_s} \subset \mathcal{V}$, the set ${\mathcal{V}_s}$ is r-reachable if there exists a node $i \in {\mathcal{V}_s}$ such that  $\left| {{\mathcal{N}_i}\backslash {\mathcal{V}_s}} \right| \geqslant r$, where $r \in \mathbb{Z}_{\geqslant 0}$.\hfill $\square$
\smallskip

\noindent
\textbf{Definition 7 (r-robust graph) \cite{Wupaper}.} A directed graph $\mathcal{G}$ is called an r-robust graph with  $r \in \mathbb{Z}_{\geqslant 0}$ if for every pair of nonempty, disjoint subsets of $\mathcal{V}$, at least one of the subsets is r-reachable.\hfill $\square$
\smallskip

\noindent
\textbf{Assumption 3.} If at most $q$ neighbors of each intact agents is under attack, at least $(q + 1)$ neighbors of each intact agents are intact.\par
\smallskip

\noindent
\textbf{Lemma 6.} \cite{Wupaper} Consider an r-robust time-varying directed graph $\mathcal{G}(t)$. Then, the graph has a directed spanning tree, if and only if $\mathcal{G}(t)$ is 1-robust. \par

\begin{theorem}
{Consider the DMAS  \eqref{eq1} under attack with the proposed resilient control protocol $\tilde{u}_{i}$ in \eqref{eq33a}. Let the time-varying graph $\mathcal{G}(t)$ be such that at each time instant $t$, Assumption $1$ and Assumption $3$ are satisfied. Then, ${\mathop {\lim }\limits_{t \to \infty } ||x_j(t)-x_{i}(t)|| = 0}\,\,\,\,\,{\forall i,j \in {\mathcal{N}_{Int}}}$.}
\end{theorem}

\begin{proof}
The DMAS  \eqref{eq1}  with the proposed resilient control protocol $\tilde{u}_{i}$ in \eqref{eq33a} without noise can be written as
\begin{equation}
{\dot x _i} = A{x _i} + cBK \sum\limits_{j \in {N_i}} {{a_{ij}}(t)\left( {{x_j} - {x_i}} \right)}
\label{m34}
\end{equation}
where ${a_{ij}}(t) = {\Psi _{ij}}(t){C_j}(t){a_{ij}}$. The global form of resilient system dynamics in \eqref{m34} becomes 
\begin{equation}
\dot x =({{I_N} \otimes A} -c\mathcal{L}(t) \otimes {BK})x
\label{m35}
\end{equation}
where $\mathcal{L}(t)$ denotes the time-varying graph Laplacian matrix of the directed graph $\mathcal{G}(t)$. Based on Assumption $3$, even if $q$ neighbors of an intact agent are attacked and collude to send the corrupted value to misguide it, there still exists $q + 1$ intact neighbors that communicate values different from the compromised ones. Moreover, since at least $q+1$ of the intact agent's neighbors are intact, it can update its trust values to remove the compromised neighbors. Furthermore, since the time varying graph $\mathcal{G}(t)$ resulting from isolating the compromised agents is 1-robust, based on Definition $7$ and Lemma $6$, the entire network is still connected to the intact agents. Therefore, there exists a spanning tree in the graph associated with all intact agents $\mathcal{N}_{Int}$. Hence, it is shown in \cite{renbeard2005} that the solutions of DMAS in \eqref{m35} reach consensus on desired behavior if the time-varying graph $\mathcal{G}(t)$ jointly contains a spanning tree as the network evolves with time. This results in ${\mathop {\lim }\limits_{t \to \infty } ||x_j(t)-x_{i}(t)|| = 0}\,\,\,\,\,{\forall i,j \in {\mathcal{N}_{Int}}}$ assymptotically. This completes the proof.
\end{proof}

\smallskip

\noindent
\textbf{Remark 7.} The proposed approach discards the compromised agent only when an attack is detected, in contrast to most of the existing methods that are based on solely the discrepancy among agents. Note that discrepancy can be the result of a legitimate change in the state of one agent. Moreover, in the beginning of synchronization, there could be a huge discrepancy between agents' states that should not be discarded.\par

\section{Simulation Results}
In this section, an example is provided to illustrate the effectiveness of the proposed detection and mitigation approaches. Consider a group of 5 homogeneous agents with the dynamics defined as
\begin{equation}
  {{\dot x}_k} = {A}{x_k} + {B}{u_k} \hfill \\
\begin{array}{*{20}{c}}
  {}&{k = 1, \ldots ,5} 
\end{array}
\label{m37}
\end{equation}
where 
\[{A} = \left[ {\begin{array}{*{20}{c}}
  0&-1 \\ 
  1&0 
\end{array}} \right],\,\,\,{B} = \left[ {\begin{array}{*{20}{c}}
  1 \\ 
  0 
\end{array}}\right].\]
The communication graph is shown in Fig. 1. In the absence of an attack, agents reach desired synchronization and there emerges the healthy behavior of the system with noisy communication as shown in Fig. 2.

\begin{figure}[!ht]
\begin{center}
\includegraphics[width=1.5in,height=1.5in]{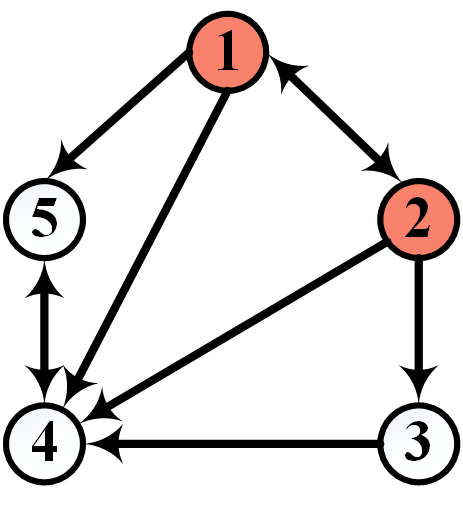}
\vspace{-5pt}\caption{Communication topology.}\label{fig:Fig2-1}
\captionsetup{justification=centering}
\end{center}
\end{figure}

\begin{figure}
    \centering
    \begin{subfigure}[b]{0.46\textwidth}
        \centering
        \includegraphics[width=1\linewidth,height=3.2cm]{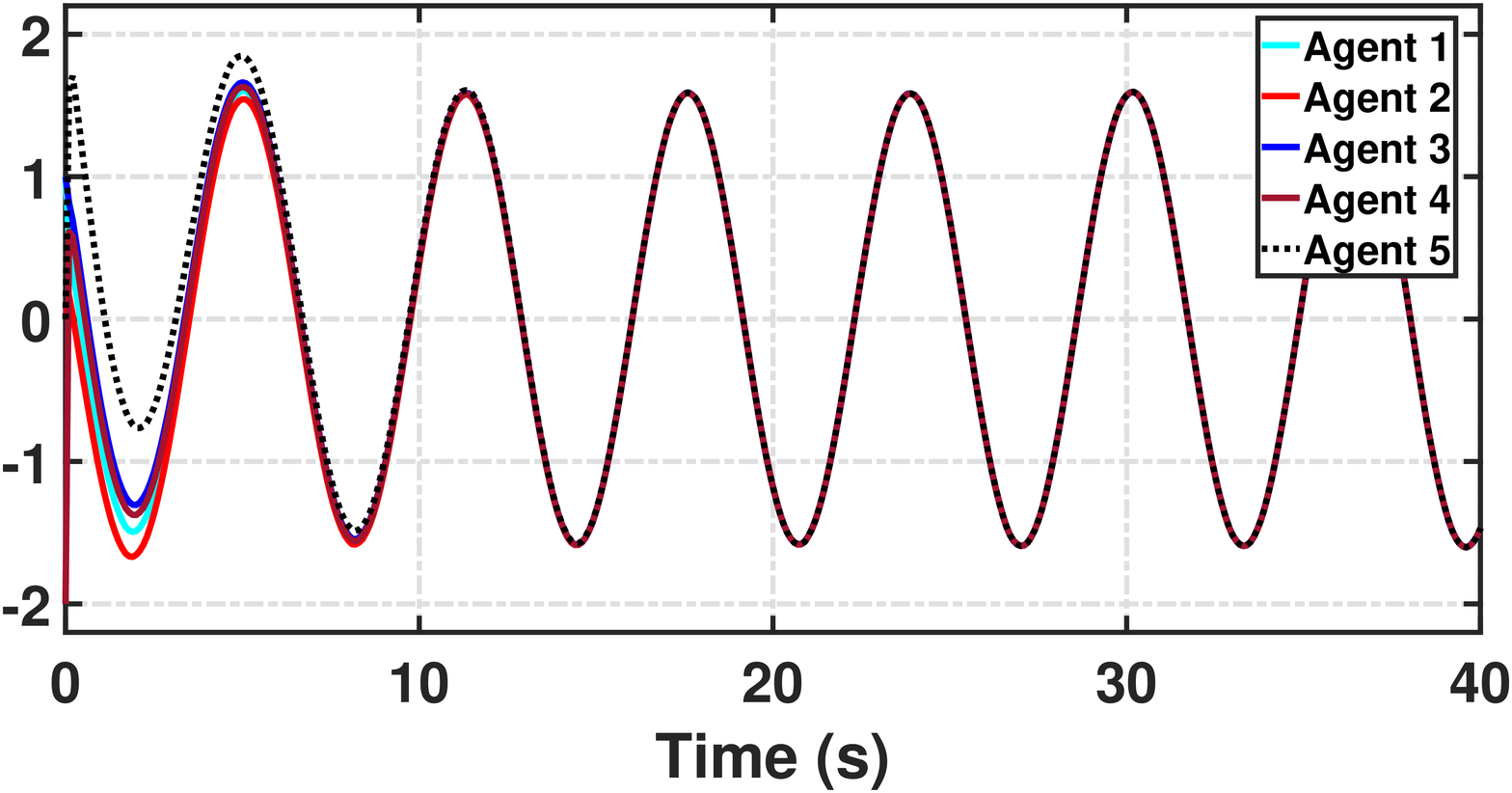}   
        \caption{}
        \label{fig:S_RS_UL1}
    \end{subfigure}
    \hspace{0cm}
    \begin{subfigure}[b]{0.46\textwidth}
        \centering
        \includegraphics[width=1\linewidth,height=3.2cm]{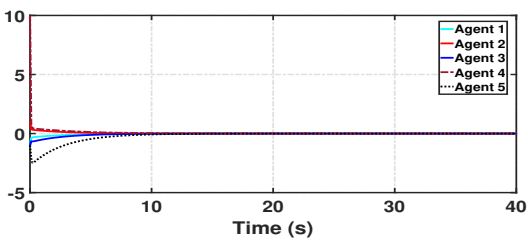}
        \caption{}
        \label{fig:S_RS_L2}
    \end{subfigure}
    \caption{Desired synchronization in the absence of attack. (a) The state of agents. (b) The local neighborhood tracking error of agents.}
    \label{fig:RS_UL2}
\end{figure}
\vspace{-0.3cm}
\subsection{IMP-based attacks}
\vspace{-0.3cm}
This subsection analyzes the effects of IMP-based attacks and illustrates our attack detection and mitigation scheme. The attack signal is assumed to be $f = 20\sin (t)$. This is an IMP-based attack and is assumed to be launched on Agent 1 (root node) at time t=20. The results are shown in Fig. 3. It can be seen that the compromised agent destabilizes the entire network. This result is consistent with Theorem 2. It is shown in Fig. 4 that the same IMP-based attack on Agent 5 (noon-root node) cannot destabilize the entire network.  However, Agent 4, which is the only agent reachable from Agent 5, does not synchronize to the desired consensus trajectory. Moreover, one can see that the local neighborhood tracking error converges to zero for all intact agents except the compromised Agent 5. These results are in line with Theorem 3. Fig.5 shows the divergence in the presence of non-IMP based attack on Agent 5 based on Theorem 4. Then, the effect of attack is rejected using the presented belief-based detection and mitigation approach in Theorem 4 and Theorem 6. Fig.6 shows that reachable agents follow the desired consensus trajectory, even in the presence of the attack. 
\begin{figure}[!ht]
\begin{center}
\includegraphics[width=1\linewidth,height=3.2cm]{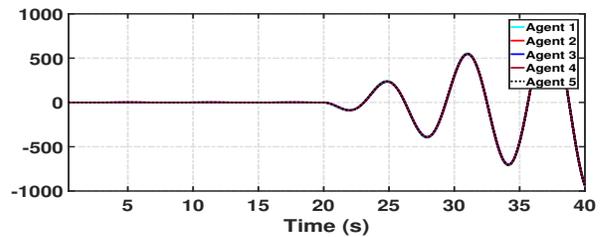}
\vspace{-5pt}\caption{The state of agents when Agent 1 is under an IMP-based attack.}\label{fig:Fig2-2}
\captionsetup{justification=centering}
\end{center}
\end{figure}

\begin{figure}
    \centering
    \begin{subfigure}[b]{0.46\textwidth}
        \centering
        \includegraphics[width=1\linewidth,height=3.2cm]{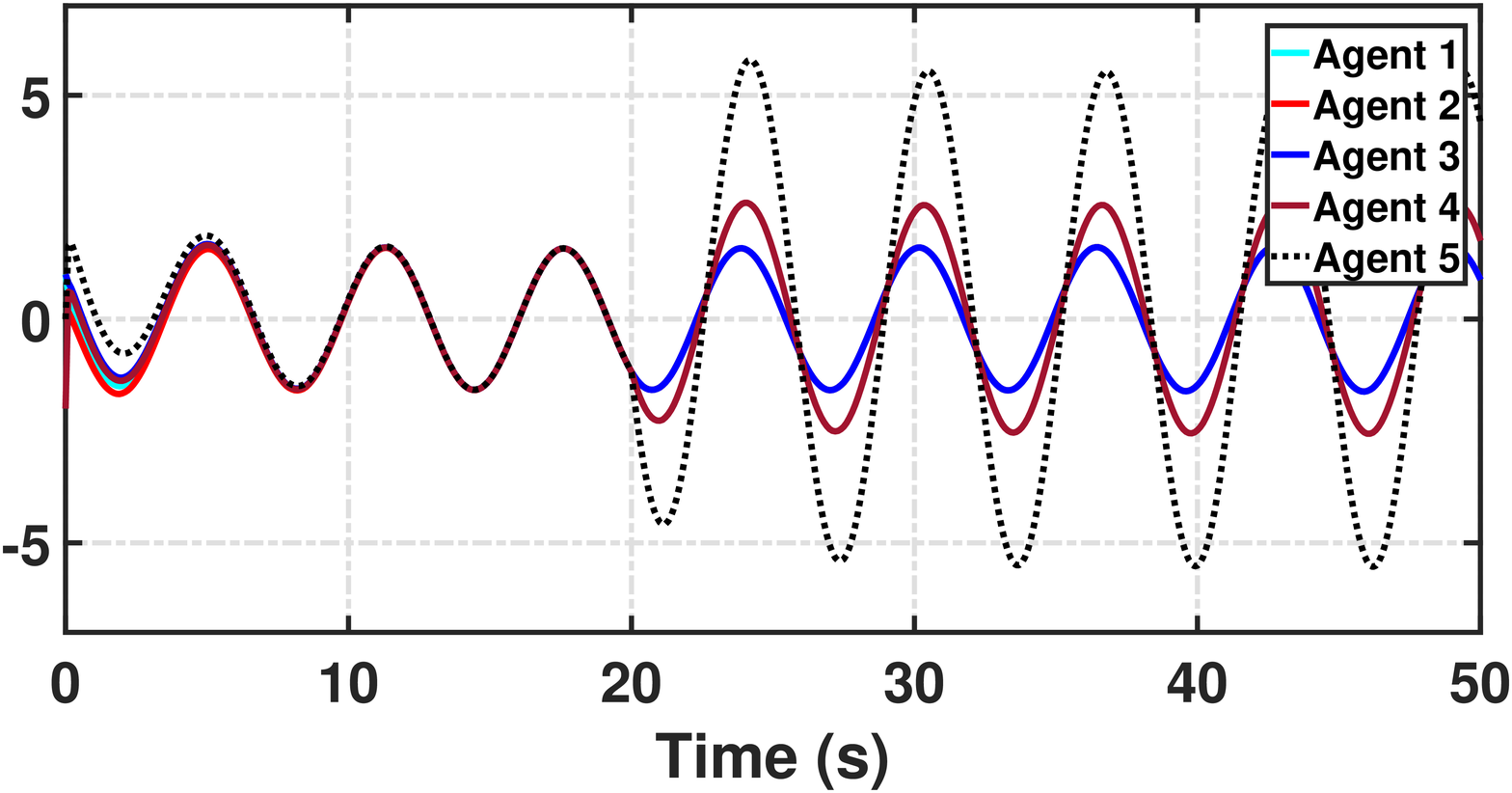}   
        \caption{}
        \label{fig:S_RS_UL3}
    \end{subfigure}
    \hspace{0cm}
    \begin{subfigure}[b]{0.46\textwidth}
        \centering
        \includegraphics[width=1\linewidth,height=3.2cm]{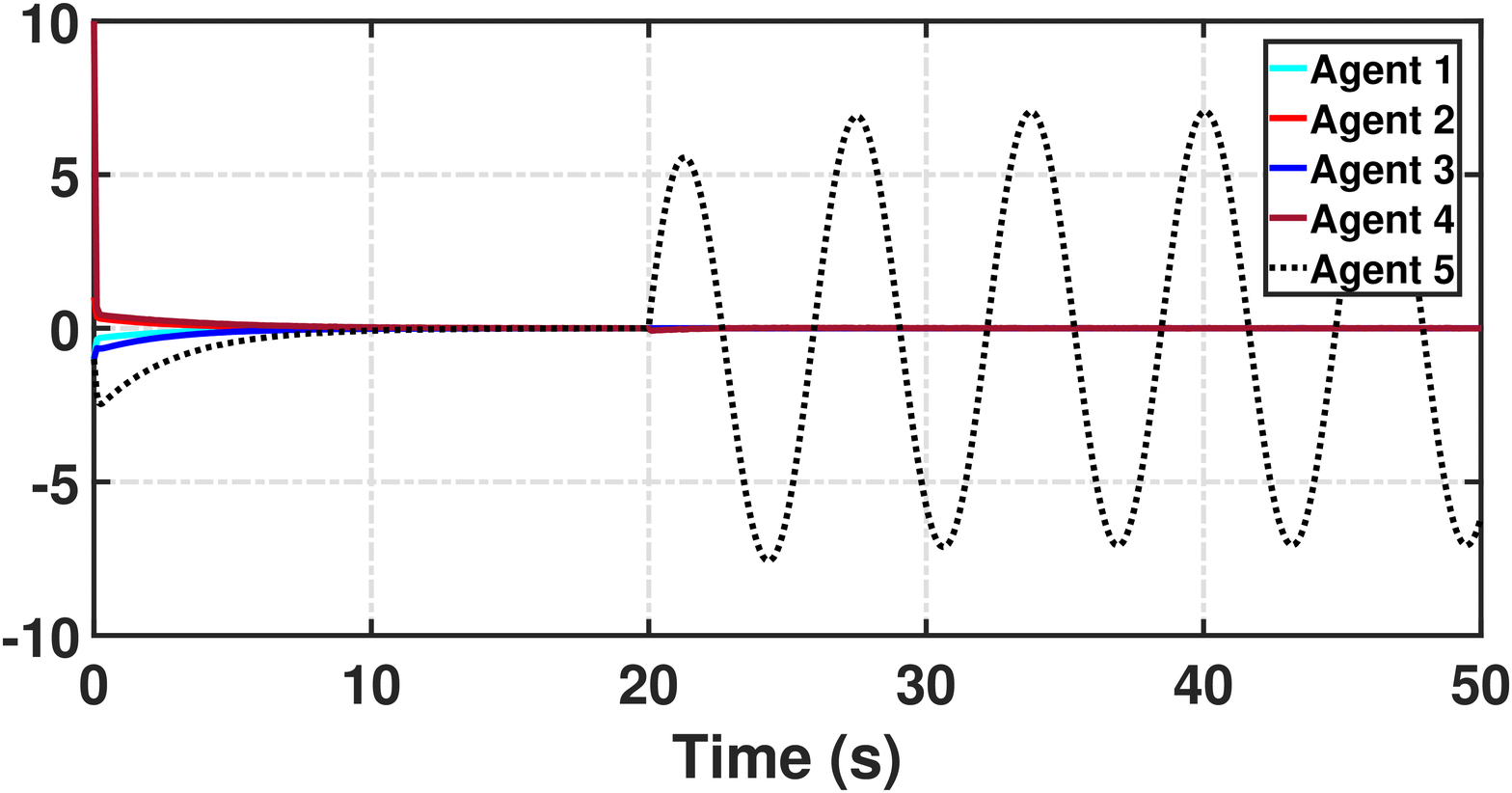}
        \caption{}
        \label{fig:S_RS_L4}
    \end{subfigure}
    \caption{Agent 5 is under IMP-based attack. (a) The state of agents. (b) The local neighborhood tracking error of agents.}
    \label{fig:RS_UL1}
\end{figure}

\begin{figure}[!ht]
\begin{center}
\includegraphics[width=1\linewidth,height=3.2cm]{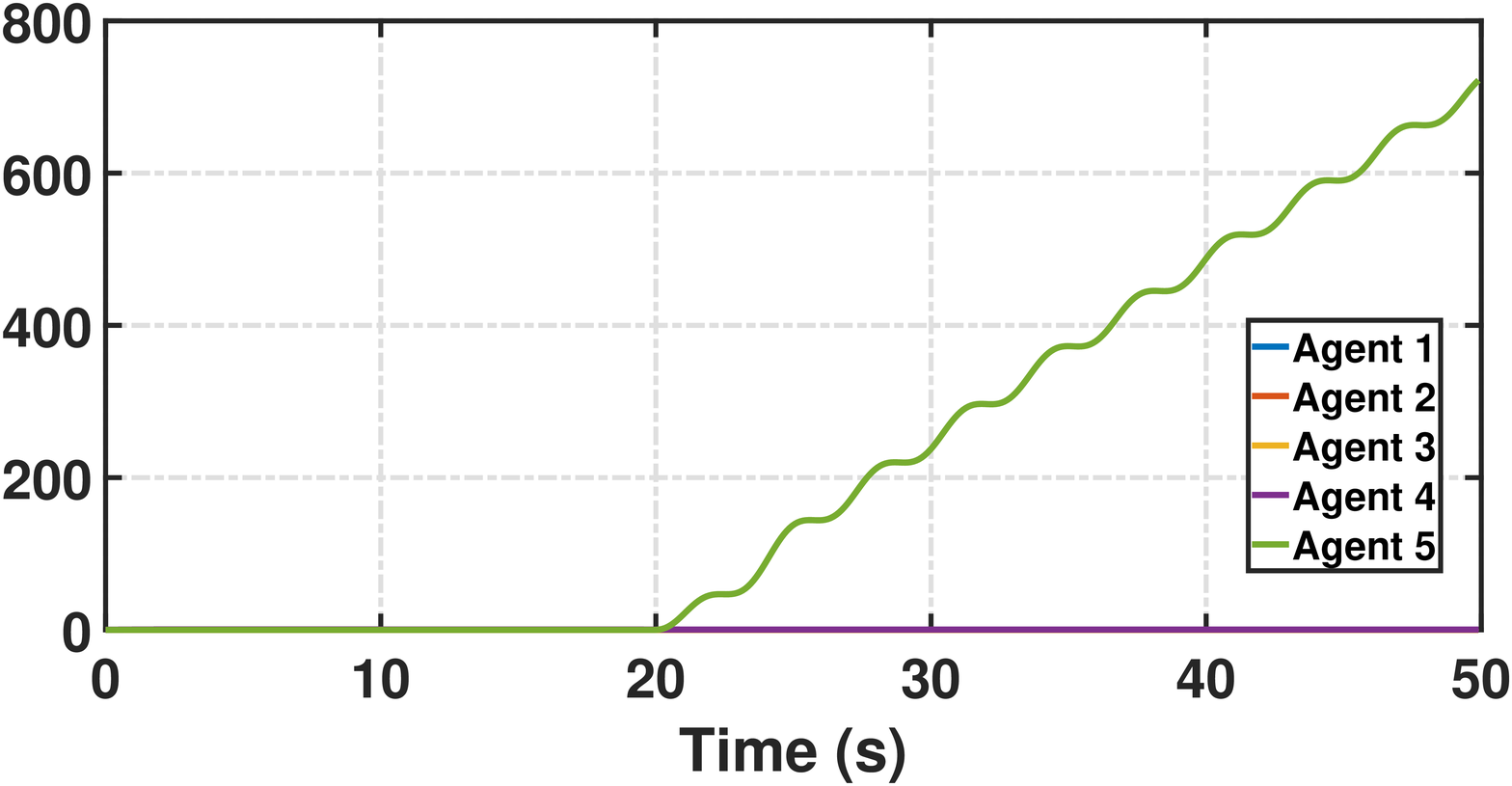}
\vspace{-5pt}\caption{Divergence for state of agents when Agent 5 is under a IMP-based attack.}\label{fig:Fig2-3}
\captionsetup{justification=centering}
\end{center}
\end{figure}
\vspace{-0.3cm}

\begin{figure}[!ht]
\begin{center}
\includegraphics[width=1\linewidth,height=3.2cm]{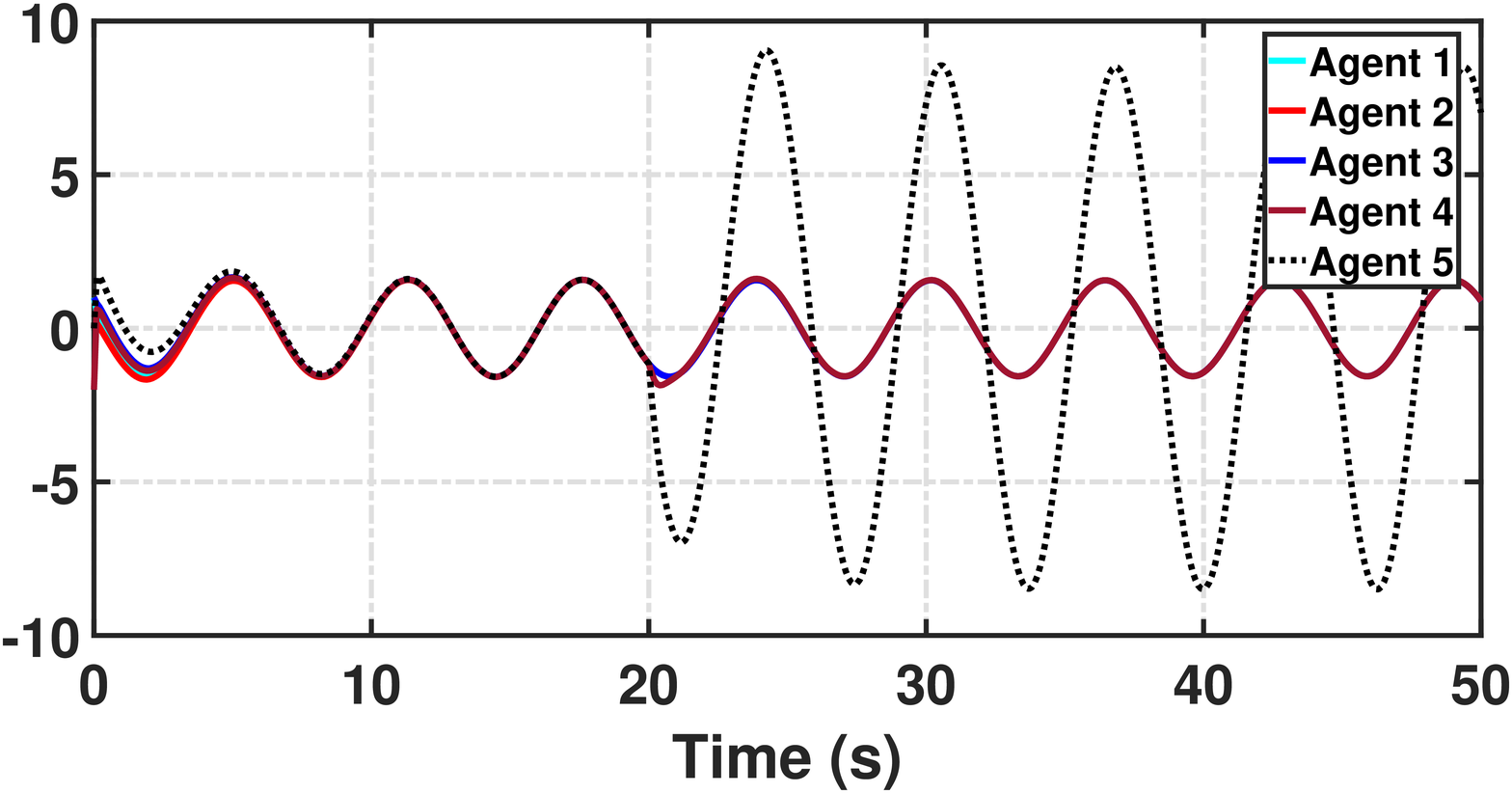}
\vspace{-5pt}\caption{The state of agents using the proposed attack detection and mitigation  approach for IMP-based attack.}\label{fig:Fig2-3}
\captionsetup{justification=centering}
\end{center}
\end{figure}
\vspace{-0.3cm}
\subsection{Non-IMP-based attacks}
This subsection analyzes the effects of non-IMP-based attacks and validates our attack detection and mitigation approach. The attack signal is assumed to be $f = 10 + 5\sin (2t)$. The effect of this non-IMP-based attack on Agent 5 (non-root node) is shown in Fig.7. It can be seen that this non-IMP-based attack on Agent 5 only affects the reachable Agent 4. Then, Fig.8 shows the divergence in the presence of non-IMP based attack on Agent 5 based on Theorem 5.  It is shown in Fig.9 that the effect of the attack is removed for the intact Agent $4$ using belief-based detection and mitigation approaches presented in Theorems 5 and 6. 

\begin{figure}[!ht]
\begin{center}
\includegraphics[width=1\linewidth,height=3.2cm]{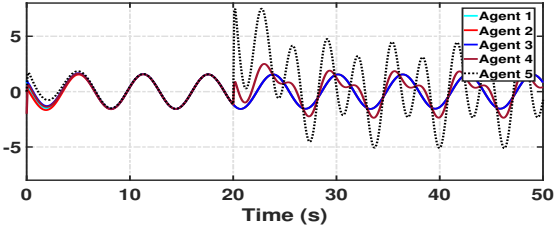}
\vspace{-5pt}\caption{The state of agents when Agent 5 is under a non-IMP-based attack.}\label{fig:Fig2}
\captionsetup{justification=centering}
\end{center}
\end{figure}

\begin{figure}[!ht]
\begin{center}
\includegraphics[width=1\linewidth,height=3.2cm]{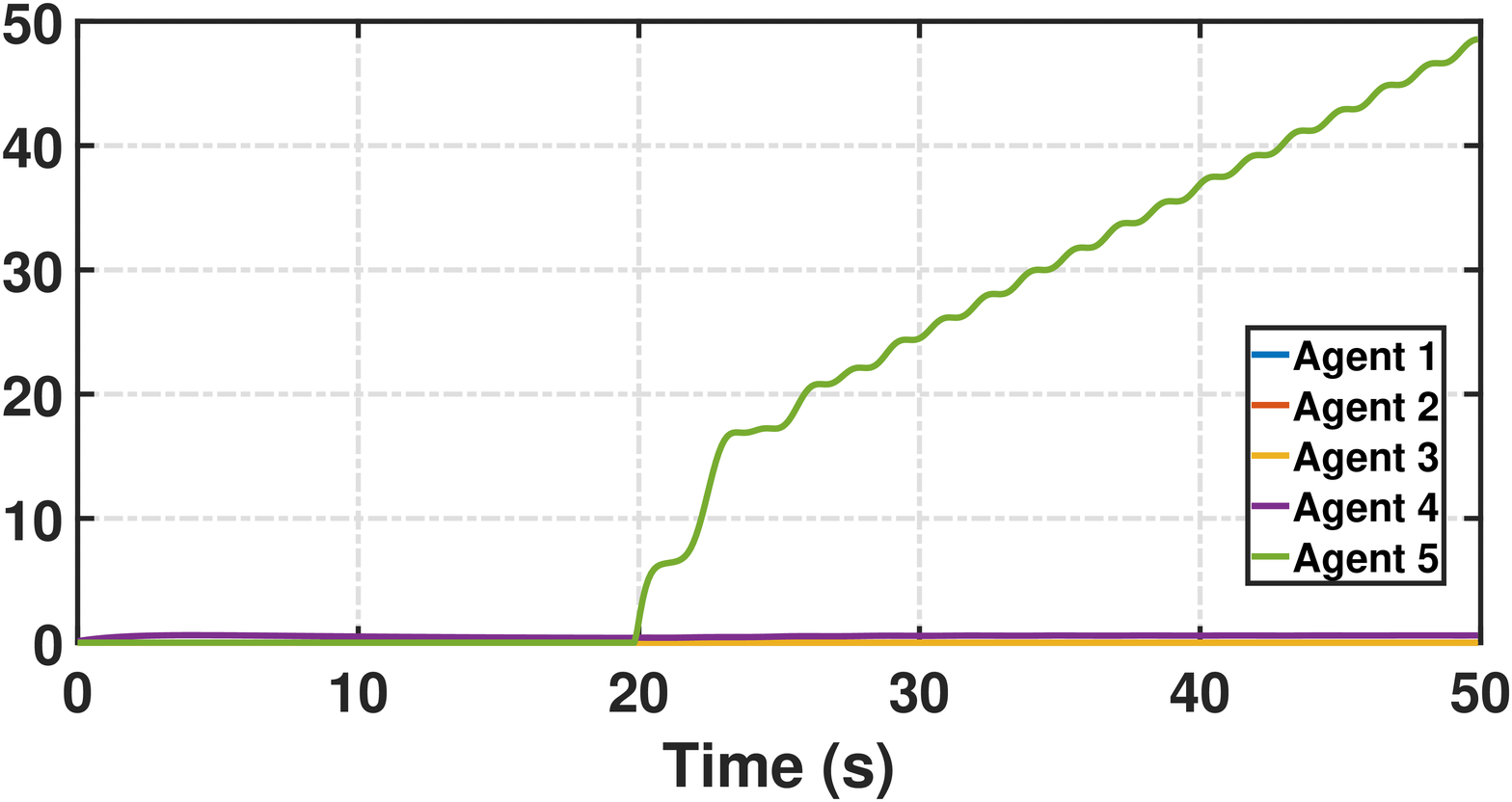}
\vspace{-5pt}\caption{Divergence for state of agents when Agent 5 is under a non-IMP based attack.}\label{fig:Fig2-4}
\captionsetup{justification=centering}
\end{center}
\end{figure}

\begin{figure}[!ht]
\begin{center}
\includegraphics[width=1\linewidth,height=3.2cm]{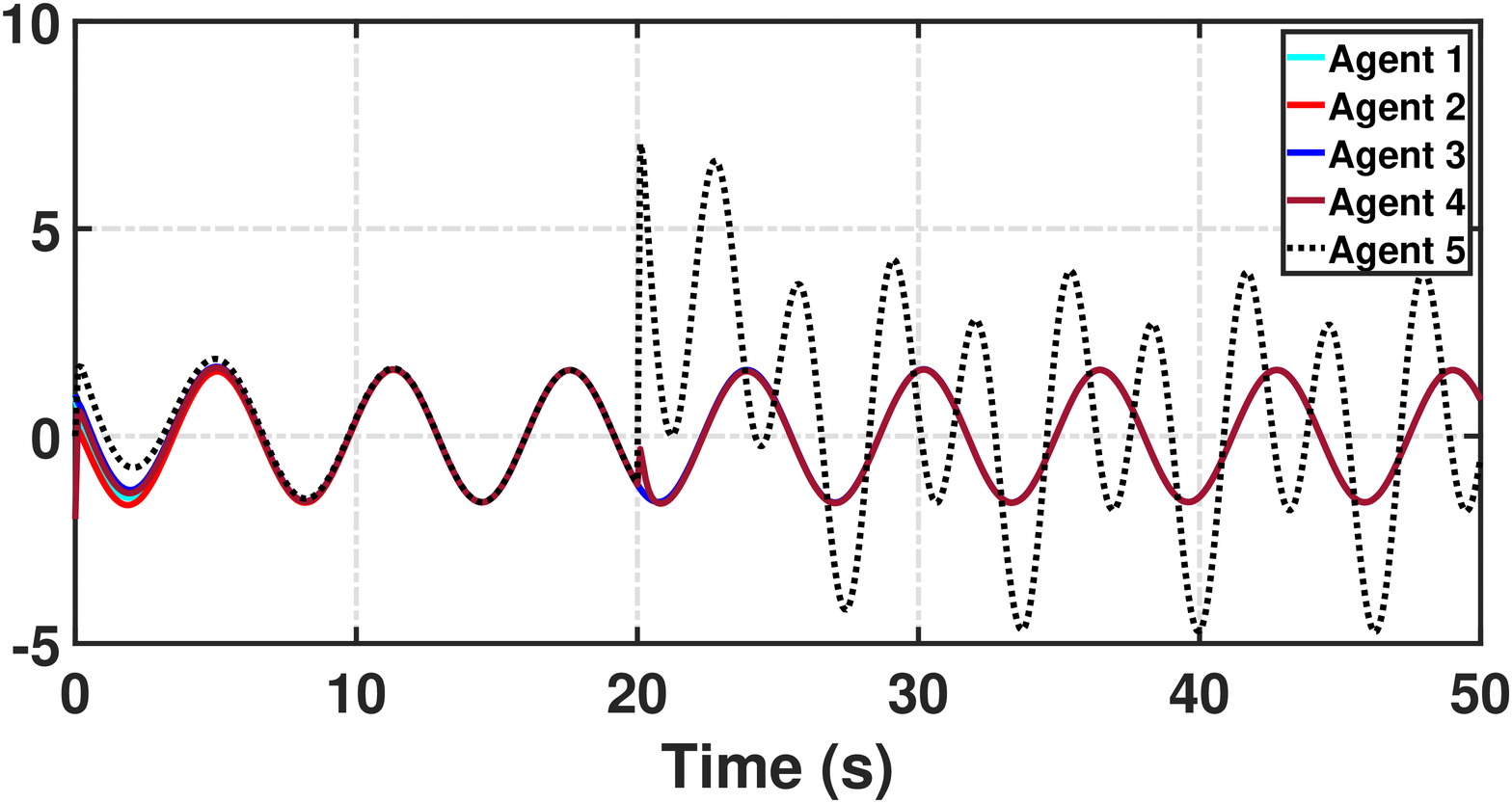}
\vspace{-5pt}\caption{The state of agents after attack detection and mitigation for non-IMP based attack.}\label{fig:Fig2-4}
\captionsetup{justification=centering}
\end{center}
\end{figure}

\section{Conclusion}
A resilient control framework has been introduced for DMASs. First, the effects of IMP-based and non-IMP-based attacks on DMASs have been analyzed using a graph-theoretic approach. Then, a KL divergence based criterion, using only the observed local information of agents, has been employed to detect attacks. Each agent detects its neighbors' misbehaviors, consequently forming a self-belief about the correctness of its own information, and continuously updates its self-belief and communicates it with its neighbors to inform them about the significance of its outgoing information. Additionally, if the self-belief value of an agent is low, it forms beliefs on the type of its neighbors (intact or compromised) and, consequently, updates its trust of its neighbors. Finally, agents incorporate their neighbors' self-beliefs and their own trust values in their control protocols to slow down and mitigate attacks. 

A possible direction for future work is to extend these results to synchronization of DMASs with nonlinear dynamics. Since nonlinear systems can exhibit finite-time escape behavior, a problem of interest is to find the conditions under which the attacker can make the trajectories of agents become unbounded in finite time, and to obtain detection and mitigation mechanisms to counteract such attacks fast and thus avoid instability.

\appendices 
\section{Proof of Lemma 5}  
Using \eqref{m6}, the KL divergence between error sequences ${\varphi _i}$ and ${\tau _i}$ can be written as 
\begin{equation}
{D_{KL}}({\varphi _i}||{\tau _i}) = {\mathbb{E}_1}[\log {P_{{\varphi _i}}} - \log {P_{{\tau _i}}}] 
\label{m40}
\end{equation}
where probability density functions ${P_{{\varphi _i}}}$ and ${P_{{\tau _i}}}$ are defined in \eqref{m4}. Using \eqref{m4}, \eqref{m40} becomes
\begin{align}
\begin{gathered}
 {D_{KL}}({\varphi _i}||{\tau _i})  \hfill \\= {\mathbb{E}_1}[\log \left( {\frac{1}{{\sqrt {2\pi } \left| {{\sigma _{1i}}} \right|}}{e^{ - \frac{{{{({q_i} - {\mu _{1i}})}^2}}}{{2\sigma _{1i}^2}}}} + \frac{1}{{\sqrt {2\pi } \left| {{\sigma _{1i}}} \right|}}{e^{ - \frac{{{{({q_i} + {\mu _{1i}})}^2}}}{{2\sigma _{1i}^2}}}}} \right) \hfill \\
   - \log \left( {\frac{1}{{\sqrt {2\pi } \left| {{\sigma _{2i}}} \right|}}{e^{ - \frac{{{{({q_i} - {\mu _{2i}})}^2}}}{{2\sigma _{2i}^2}}}} + \frac{1}{{\sqrt {2\pi } \left| {{\sigma _{2i}}} \right|}}{e^{ - \frac{{{{({q_i} + {\mu _{2i}})}^2}}}{{2\sigma _{2i}^2}}}}} \right)] \hfill \\ 
\end{gathered}
\label{m41}
\end{align}

By the aid of the logarithm property as $\log (a + b) = \log (a) + \log (1 + b/a)$, \eqref{m41} turns into
\begin{align}
\begin{gathered}
  {D_{KL}}({\varphi _i}||{\tau _i}) =  \hfill \\
  \, = \underbrace {{\mathbb{E}_1}[\log \left( {\frac{1}{{\sqrt {2\pi } \left| {{\sigma _{1i}}} \right|}}{e^{ - \frac{{{{({q_i} - {\mu _{1i}})}^2}}}{{2\sigma _{1i}^2}}}}} \right) - \log \left( {\frac{1}{{\sqrt {2\pi } \left| {{\sigma _{2i}}} \right|}}{e^{ - \frac{{{{({q_i} - {\mu _{2i}})}^2}}}{{2\sigma _{2i}^2}}}}} \right)]}_{{T_1}} \hfill \\
  \,\,\,\,\,\,\,\,\,\,\,\,\,\,\,\,\,\,\,\,\, + \underbrace {{\mathbb{E}_1}[\log \left( {1 + {e^{ - \frac{{2{q_i}{\mu _{1i}}}}{{\sigma _{1i}^2}}}}} \right) - \log \left( {1 + {e^{ - \frac{{2{q_i}{\mu _{2i}}}}{{\sigma _{2i}^2}}}}} \right)]}_{{T_2}} \hfill \\ 
\end{gathered} 
\label{m42}
\end{align}

The first term in \eqref{m42} is a KL divergence formula for statistical sequences with normal Gaussian distribution which is given in \cite{KLBOOK} as
\begin{equation}
{T_1} = \frac{1}{2}\left( {\log \frac{{\sigma _{2i}^2}}{{\sigma _{1i}^2}} - 1 + (\sigma _{2i}^{ - 2}\sigma _{1i}^2)} \right) + \frac{1}{2}\sigma _{2i}^{ - 2}{({\mu _{2i}} - {\mu _{1i}})^2}
\label{m43}
\end{equation}

The second term ${T_2}$ in \eqref{m42}, using power series expansion $\log (1 + a) = \sum\limits_{n \geqslant 0} {\left( {{{{{( - 1)}^n}{a^{n + 1}}} \mathord{\left/
 {\vphantom {{{{( - 1)}^n}{a^{n + 1}}} {\left( {n + 1} \right)}}} \right.
 \kern-\nulldelimiterspace} {\left( {n + 1} \right)}}} \right)}$ and ignoring higher order terms, can be approximated as
\begin{equation}
{T_2} \approx {\mathbb{E}_1}[{e^{ - \frac{{2{q_i}{\mu _{1i}}}}{{\sigma _{1i}^2}}}} - \frac{{{{({e^{ - \frac{{2{q_i}{\mu _{1i}}}}{{\sigma _{1i}^2}}}})}^2}}}{2}] - {\mathbb{E}_1}[{e^{ - \frac{{2{q_i}{\mu _{2i}}}}{{\sigma _{2i}^2}}}} - \frac{{{{({e^{ - \frac{{2{q_i}{\mu _{2i}}}}{{\sigma _{2i}^2}}}})}^2}}}{2}]
\label{m44}
\end{equation}
which can be expressed as
\begin{align}
\begin{gathered}
  {T_2} \approx \int\limits_{ - \infty }^\infty  {{P_{{\varphi _i}}}{e^{ - \frac{{2{q_i}{\mu _{1i}}}}{{\sigma _{1i}^2}}}}} d{q_i} - \frac{1}{2}\int\limits_{ - \infty }^\infty  {{P_{{\varphi _i}}}{e^{ - \frac{{4{q_i}{\mu _{1i}}}}{{\sigma _{1i}^2}}}}} d{q_i} \hfill \\
   - \int\limits_{ - \infty }^\infty  {{P_{{\varphi _i}}}{e^{ - \frac{{2{q_i}{\mu _{2i}}}}{{\sigma _{2i}^2}}}}} d{q_i} + \frac{1}{2}\int\limits_{ - \infty }^\infty  {{P_{{\varphi _i}}}{e^{ - \frac{{4{q_i}{\mu _{2i}}}}{{\sigma _{2i}^2}}}}} d{q_i} \hfill \\ 
\end{gathered} 
\label{m45}
\end{align}

Now, the first term of ${T_2}$ can be written as
\begin{align}
\begin{gathered}
 \int\limits_{ - \infty }^\infty  {{P_{{\varphi _i}}}{e^{ - \frac{{2{q_i}{\mu _{1i}}}}{{\sigma _{1i}^2}}}}} d{q_i} \hfill \\
   = \int\limits_{ - \infty }^\infty  {\frac{1}{{\sqrt {2\pi } \left| {{\sigma _{1i}}} \right|}}{e^{ - \frac{{{{({q_i} + {\mu _{1i}})}^2}}}{{2\sigma _{1i}^2}}}}d{q_i} + \int\limits_{ - \infty }^\infty  {\frac{1}{{\sqrt {2\pi } \left| {{\sigma _{1i}}} \right|}}{e^{ - \frac{{{{({q_i} + {\mu _{1i}})}^2} + 4{q_i}{\mu _{1i}}}}{{2\sigma _{1i}^2}}}}d{q_i}} }  \hfill \\ 
\end{gathered} 
\label{m46}
\end{align}

Using the fact that density integrates to 1, \eqref{m46} becomes
\begin{equation}
\int\limits_{ - \infty }^\infty  {{P_{{\varphi _i}}}{e^{ - \frac{{2{q_i}{\mu _{1i}}}}{{\sigma _{1i}^2}}}}} d{q_i} = 1 + {e^{\frac{{4\mu _{1i}^2}}{{\sigma _{1i}^2}}}}
\label{m47}
\end{equation}

Similarly, second term of ${T_2}$ can be written as
\begin{align}
\begin{gathered}
- \frac{1}{2}\int\limits_{ - \infty }^\infty  {{P_{{\varphi _i}}}{e^{ - \frac{{4{q_i}{\mu _{1i}}}}{{\sigma _{1i}^2}}}}} d{q_i} \hfill \\
   =  - \frac{1}{{2\sqrt {2\pi } \left| {{\sigma _{1i}}} \right|}}\int\limits_{ - \infty }^\infty  {\left( {{e^{ - \frac{{{{({q_i} + 3{\mu _{1i}})}^2} - 8\mu _{1i}^2}}{{2\sigma _{1i}^2}}}}d{q_i} + {e^{ - \frac{{{{({q_i} + 5{\mu _{1i}})}^2} - 24\mu _{1i}^2}}{{2\sigma _{1i}^2}}}}d{q_i}} \right)}  \hfill \\ 
\end{gathered} 
\label{m48}
\end{align}
which yields
\begin{equation}
- \frac{1}{2}\int\limits_{ - \infty }^\infty  {{P_{{\varphi _i}}}{e^{ - \frac{{4{q_i}{\mu _{1i}}}}{{\sigma _{1i}^2}}}}} d{q_i} =  - \frac{1}{2}\left( {{e^{\frac{{4\mu _{1i}^2}}{{\sigma _{1i}^2}}}} + {e^{\frac{{12\mu _{1i}^2}}{{\sigma _{1i}^2}}}}} \right)
\label{m49}
\end{equation}

The third term of ${T_2}$  is
\begin{align}
\begin{gathered}
   - \int\limits_{ - \infty }^\infty  {{P_{{\varphi _i}}}{e^{ - \frac{{2{q_i}{\mu _{1i}}}}{{\sigma _{1i}^2}}}}} d{q_i} \hfill \\
   =  - \frac{1}{{\sqrt {2\pi } \left| {{\sigma _{1i}}} \right|}}\int\limits_{ - \infty }^\infty  {\left( {{e^{ - \frac{{{{({q_i} - {\mu _{1i}})}^2}}}{{2\sigma _{1i}^2}}}}{e^{ - \frac{{2{q_i}{\mu _{2i}}}}{{\sigma _{2i}^2}}}} + {e^{ - \frac{{{{({q_i} + {\mu _{1i}})}^2}}}{{2\sigma _{1i}^2}}}}{e^{ - \frac{{2{q_i}{\mu _{2i}}}}{{\sigma _{2i}^2}}}}} \right)} d{q_i} \hfill \\
\end{gathered} 
\label{m50}
\end{align}
which can be written in the form 
\begin{align}
\begin{gathered}
  - \int\limits_{ - \infty }^\infty  {{P_{{\varphi _i}}}{e^{ - \frac{{2{q_i}{\mu _{1i}}}}{{\sigma _{1i}^2}}}}} d{q_i} \hfill \\
   =  - \left( {{e^{ - \frac{{\mu _{1i}^2 - \rho _1^2}}{{2\sigma _{1i}^2}}}}\int\limits_{ - \infty }^\infty  {\frac{1}{{\sqrt {2\pi } \left| {{\sigma _{1i}}} \right|}}{e^{ - \frac{{{{({q_i} - {\rho _1})}^2}}}{{2\sigma _{1i}^2}}}}d{q_i}} } \right. \hfill \\
  \left. { + {e^{ - \frac{{\mu _{1i}^2 - \rho _2^2}}{{2\sigma _{1i}^2}}}}\int\limits_{ - \infty }^\infty  {\frac{1}{{\sqrt {2\pi } \left| {{\sigma _{1i}}} \right|}}{e^{ - \frac{{{{({q_i} - {\rho _2})}^2}}}{{2\sigma _{1i}^2}}}}d{q_i}} } \right) \hfill \\ 
\end{gathered} 
\label{m51}
\end{align}
where ${\rho _1} = ({\mu _{1i}} - 2{\mu _{2i}}\sigma _{1i}^2\sigma _{2i}^{ - 2})$ and ${\rho _2} = ({\mu _{1i}} + 2{\mu _{2i}}\sigma _{1i}^2\sigma _{2i}^{ - 2})$ which becomes
\begin{align}
\begin{gathered}
- \int\limits_{ - \infty }^\infty  {{P_{{\varphi _i}}}{e^{ - \frac{{2{q_i}{\mu _{1i}}}}{{\sigma _{1i}^2}}}}} d{q_i} =  - \left( {{e^{ - \frac{{\mu _{1i}^2 - \rho _1^2}}{{2\sigma _{1i}^2}}}} + {e^{ - \frac{{\mu _{1i}^2 - \rho _2^2}}{{2\sigma _{1i}^2}}}}} \right)
\end{gathered} 
\label{m52}
\end{align}

Similarly, the last term of ${T_2}$ can be simplified as 
\begin{align}
\begin{gathered}
\frac{1}{2}\int\limits_{ - \infty }^\infty  {{P_{{\varphi _i}}}{e^{ - \frac{{2{q_i}{\mu _{1i}}}}{{\sigma _{1i}^2}}}}} d{q_i} = \frac{1}{2}\left( {{e^{ - \frac{{\mu _{1i}^2 - \rho _3^2}}{{2\sigma _{1i}^2}}}} + {e^{ - \frac{{\mu _{1i}^2 - \rho _4^2}}{{2\sigma _{1i}^2}}}}} \right)
\end{gathered} 
\label{m53}
\end{align}
where ${\rho _3} = ({\mu _{1i}} - 4{\mu _{2i}}\sigma _{1i}^2\sigma _{2i}^{ - 2})\,$ and ${\rho _4} = ({\mu _{1i}} + 4{\mu _{2i}}\sigma _{1i}^2\sigma _{2i}^{ - 2})\,$. Adding \eqref{m47}, \eqref{m49}, \eqref{m52} and \eqref{m53}, ${T_2}$ can be written as
\begin{align}
\begin{gathered}
  {T_2} \approx {e^{ - \frac{{\mu _{1i}^2}}{{2\sigma _{1i}^2}}}}\left( {\frac{1}{2}\left( {{e^{\frac{{\rho _3^2}}{{2\sigma _{1i}^2}}}} + {e^{\frac{{\rho _4^2}}{{2\sigma _{1i}^2}}}}} \right) - \left( {{e^{\frac{{\rho _1^2}}{{2\sigma _{1i}^2}}}} + {e^{\frac{{\rho _2^2}}{{2\sigma _{1i}^2}}}}} \right)} \right) \hfill \\
   + 1 + \frac{1}{2}{e^{\frac{{4\mu _{1i}^2}}{{\sigma _{1i}^2}}}}\left( {1 - {e^{\frac{{8\mu _{1i}^2}}{{\sigma _{1i}^2}}}}} \right) \hfill \\ 
\end{gathered} 
\label{m54}
\end{align}

Now, using \eqref{m43}-\eqref{m44} and \eqref{m54}, one gets \eqref{m6}. This completes the proof. 

\bibliographystyle{ieeetr}
\bibliography{reference}

\addtolength{\textheight}{-3cm}
\begin{IEEEbiography}[{\includegraphics[width=1in,height=1.25in,clip,keepaspectratio]{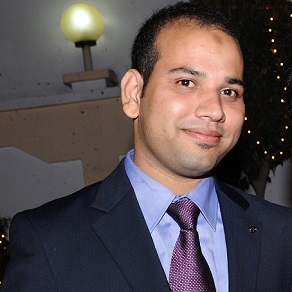}}]{{\bf Aquib Mustafa}}
 (S’17) received the B. Tech. degree from the Aligarh Muslim University, Aligarh, India, in 2013, and the Master’s degree from the Indian Institute of Technology Kanpur, Kanpur, India, in 2016. He is currently pursuing the Ph.D. degree in the Department of Mechanical Engineering, Michigan State University, East Lansing, USA. His primary research interests include Resilient control, Multi-agent systems, and sensor networks.
\end{IEEEbiography}

\begin{IEEEbiography}[{\includegraphics[width=1in,height=1.25in,clip,keepaspectratio]{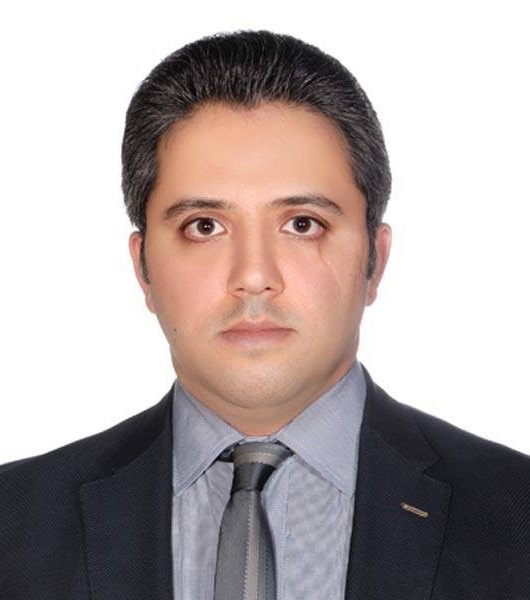}}]{{\bf Rohollah Moghadam}}
 (S’17) received the M.S. degree from the Shahrood University of Technology, Shahrood, Iran, in 2007, in electrical engineering. He was a visiting scholar with the University of Texas at Arlington Research Institute, Fort Worth, TX, USA in 2016. He is currently pursuing the Ph.D. degree in the Department of Electrical and Computer Engineering, Missouri University of Science and Technology, Rolla, USA.  His current research interests include cyber-physical systems, reinforcement learning, neural network, network control systems, and distributed control of multi-agent systems. \end{IEEEbiography}

\begin{IEEEbiography}[{\includegraphics[width=1in,height=1.25in,clip,keepaspectratio]{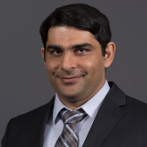}}]{{\bf Hamidreza Modares}}
(M’15) received the B.Sc. degree from Tehran University, Tehran, Iran, in 2004, the M.Sc. degree from the Shahrood University of Technology, Shahrood, Iran, in 2006, and the Ph.D. degree from the University of Texas at Arlington (UTA), Arlington, TX, USA, in 2015. From 2006 to 2009, he was with the Shahrood University of Technology as a Senior Lecturer. From 2015 to 2016, he was a Faculty Research Associate with UTA. From 2016 to 2018, he was an Assistant Professor with the Department of Electrical and Computer Engineering, Missouri University of Science and Technology, Rolla, USA. 

He is currently an Assistant Professor with the Department of Mechanical Engineering, Michigan State University, East Lansing, USA. He has authored several journal and conference papers on the design of optimal controllers using reinforcement learning. His current research interests include cyber-physical systems, machine learning, distributed control, robotics, and renewable energy microgrids.

Dr. Modares was a recipient of the Best Paper Award from the 2015 IEEE International Symposium on Resilient Control Systems, the Stelmakh Outstanding Student Research Award from the Department of Electrical Engineering, UTA, in 2015, and the Summer Dissertation Fellowship from UTA, in 2015. He is an Associate Editor of the IEEE TRANSACTIONS ON NEURAL NETWORKS AND LEARNING SYSTEMS.
\end{IEEEbiography}

\end{document}